\documentclass[prx,nofootinbib,twocolumn,showkeys,superscriptaddress,preprintnumbers,floatfix]{revtex4-1}

\usepackage{lineno}
\usepackage{algorithm}
\usepackage{algpseudocode}
\usepackage{etex}
\usepackage{ifpdf}
\usepackage{hyperref}
\usepackage{dcolumn}
\usepackage{url}
\usepackage{amsmath}
\usepackage{amscd}
\usepackage{amsfonts}
\usepackage{amssymb}
\usepackage{bm}   
\usepackage{bbm}
\usepackage{verbatim}
\usepackage{stmaryrd}
\usepackage{amsthm}
\usepackage{xcolor}
\usepackage{setspace}
\usepackage{braket}
\usepackage{empheq}
\usepackage{perpage} 
\MakePerPage{footnote}

\usepackage{tikz}
\usetikzlibrary{arrows}
\usetikzlibrary{automata}
\usetikzlibrary{decorations.pathreplacing}
\usetikzlibrary{positioning}
\usetikzlibrary{plotmarks}
\usetikzlibrary{calc}
\usetikzlibrary{patterns}
\usetikzlibrary{external}
\usepackage{pgfplots}
\pgfplotsset{compat=newest}

\usepackage{graphicx}

\theoremstyle{plain}    
\theoremstyle{plain}    
\theoremstyle{plain}    
\theoremstyle{plain}    
\theoremstyle{plain}    \newtheorem{The}{Theorem}
\theoremstyle{plain}    
\theoremstyle{plain}    
\theoremstyle{plain}    
\theoremstyle{plain}    
\theoremstyle{plain}    
\theoremstyle{plain}    
\theoremstyle{plain}

\newcommand{\CausalState}   { \mathcal{S} }

\newcommand{\forward}{+}
\newcommand{\reverse}{-}
\newcommand{\forwardreverse}{\pm} 

\newcommand{\FutureCausalState} { {\CausalState}^{\forward} }

\newcommand{\PastCausalState}   { {\CausalState}^{\reverse} }

\newcommand{\lastindex}[2]{
  \edef\tempa{0}
  \edef\tempb{#2}
  \ifx\tempa\tempb
    \edef\tempc{#1}
  \else
    \edef\tempa{0}
    \edef\tempb{#1}
    \ifx\tempa\tempb
      \edef\tempc{#2}
    \else
      \edef\tempc{#1+#2}
    \fi
  \fi
  \tempc
}

\newcommand{\CSjoint}[1][,]{
   \edef\tempa{:}
   \edef\tempb{#1}
   \ifx\tempa\tempb
      \ensuremath{\FutureCausalState\!#1\PastCausalState}
   \else
      \ensuremath{\FutureCausalState#1\PastCausalState}
   \fi
}

\newif\ifpm
\edef\tempa{\forwardreverse}
\edef\tempb{\pm}
\ifx\tempa\tempb
   \pmfalse
\else
   \pmtrue
\fi

\newcommand{\kB}{k_\text{B}}  
\newcommand{\avg}[1]{\left< #1 \right>}

\begin{document}
\title{Thermodynamic Overfitting and Generalization:
\\ Energetic Limits on Predictive Complexity}

\author{Alexander B. Boyd}
\email{alecboy@gmail.com}
\affiliation{School of Physics, Trinity College Dublin, Dublin 2, Ireland}
\affiliation{Trinity Quantum Alliance, Unit 16, Trinity Technology and Enterprise Centre, Pearse Street, Dublin 2, Ireland}
\affiliation{Beyond Institute for Theoretical Science, San Francisco, California, USA}

\author{James P. Crutchfield}
\email{chaos@ucdavis.edu}
\affiliation{Complexity Sciences Center and Physics Department,
University of California at Davis, One Shields Avenue, Davis, CA 95616}

\author{Mile Gu}
\email{mgu@quantumcomplexity.org}
\affiliation{Nanyang Quantum Hub, School of Physical and Mathematical Sciences, Nanyang Technological University, 637371, Singapore}
\affiliation{Centre for Quantum Technologies, National University of Singapore, 3 Science Drive 2, 117543, Singapore}
\affiliation{MajuLab, CNRS-UNS-NUS-NTU International Joint Research Unit, UMI 3654, 117543, Singapore}

\author{Felix C. Binder}
\email{quantum@felix-binder.net}
\affiliation{School of Physics, Trinity College Dublin, Dublin 2, Ireland}
\affiliation{Trinity Quantum Alliance, Unit 16, Trinity Technology and Enterprise Centre, Pearse Street, Dublin 2, Ireland}

\date{\today}

\bibliographystyle{quantum}

\begin{abstract}

Efficiently harvesting thermodynamic resources requires a precise understanding of their structure. This becomes explicit through the lens of information engines---thermodynamic engines that use information as fuel. Maximizing the work harvested using available information is a form of physically-instantiated machine learning that drives information engines to develop complex predictive memory to store an environment's temporal correlations.  We show that an information engine's complex predictive memory poses both energetic benefits and risks. While increasing memory facilitates detection of hidden patterns in an environment, it also opens the possibility of thermodynamic overfitting, where the engine dissipates additional energy in testing. To address overfitting, we introduce thermodynamic regularizers that incur a cost to engine complexity in training due to the physical constraints on the information engine. We demonstrate that regularized thermodynamic machine learning generalizes effectively. In particular, the physical constraints from which regularizers are derived improve the performance of learned predictive models. This suggests that the laws of physics jointly create the conditions for emergent complexity and predictive intelligence.

\end{abstract}

\keywords{nonequilibrium thermodynamics, Maxwell's demon, Landauer's Principle, machine learning, generalization}

\pacs{
05.70.Ln  
89.70.-a  
05.20.-y  
05.45.-a  
}

\maketitle

\section{Introduction}
Modern machine learning has made remarkable advances mimicking our understanding of biological intelligence by incorporating biology's fundamental building blocks, e.g., in the form of neural networks~\cite{dongare2012introduction}. Similarly, thermodynamics and statistical mechanics have made considerable contributions to machine learning \cite{sohl2015deep, bahri2020statistical}. Most recently, this has come to include proposals to directly use thermodynamic systems for machine learning tasks~\cite{melanson2023thermodynamic, coles2023thermodynamic}. Meanwhile, thermodynamic concepts in machine learning found their way back to elucidating biological intelligence~\cite{Fris10a,ali2022predictive}.
This cross-fertilisation falls in line with a larger concern about how biological intelligence is embodied~\cite{gupta2021embodied}. That is, the physicality of an intelligent agent, as manifested through its environment interactions, determines how it learns about the world. In this way, it has been recognised that machine learning and pattern recognition ``can be viewed as two facets of the same field''~\cite{bishop2006pattern}.

With this in mind, the following explores the thermodynamics of pattern prediction by analyzing a learning agent who utilises information provided by its environment for the aim of maximal energy extraction. This \emph{thermodynamic machine learning} (TML) implements the \emph{principle of maximum work production}~\cite{Boyd21a} that expresses how thermodynamic efficiency is tied to optimal prediction~\cite{Stil12a}. Specifically, selecting the model that harvests the most work from an information source is equivalent to performing \emph{maximum likelihood estimation} (MLE) on predictive models of data. This highlights the thermodynamic roots of \emph{computational mechanics}~\cite{Crut88a, Crut01a, Crut12a}---the information theory of time-series prediction and structural complexity. While maximum work production is by no means guaranteed in general nonequilibrium physical processes \cite{gold2019self}, the equivalence between machine learning and thermodynamic resource maximization does offer a mechanism that drives the emergence of structural complexity and intelligence.

Here, we explore the central role of model complexity in thermodynamic machine learning by showing that principles of prediction arise from simple physical principles---work maximization, autocorrection, and model initialization. Therein, a concern arises about overfitting which occurs in machine learning when a model performs well on training data, but fails to effectively predict further inputs (test data). In this case, the estimated model is overly-specific to the training data and does not generalize to further samples~\cite{ying2019overview}. Overfitting is closely tied to high model complexity, because it often happens when
the number of model parameters exceed what is justified by a limited dataset.

In the thermodynamic setting, there is the (conversely) related \emph{principle of requisite complexity} which states that an information engine must (at least) match the structural complexity of the environment to operate efficiently \cite{Boyd17a}. In contrast, we have \emph{thermodynamic overfitting} which reflects a thermodynamic cost to excessively complex models. Practically, this cost appears to prohibit learning exceedingly large models that may exhibit ``double descent''~\cite{nakkiran2021deep} where a regime of improved learning occurs as model complexity increases beyond the regime of overfitting.

Paralleling the strategy in conventional machine learning, to mitigate overfitting we turn to regularization by incurring a training penalty that reflects model complexity~\cite{tibshirani1996regression, zhang2010regularization}. Specifically, we introduce a \emph{thermodynamic regularizer}: a thermodynamic cost to model complexity that is proportional to the work that is dissipated.  Algorithmically, adding this regularizer results in Bayesian updates of a predictive model's edge-weights---the weights that control the engine's operation. We also introduce a cost to autocorrect the engine's predictive states, which arises from starting in a uniform distribution over the engine's memory states. The net result is an effective regularization function---essentially,  a new complexity measure that quantifies the degree to which different causal predictive states make different predictions.

Notably, the following derives a thermodynamically-based prediction algorithm that parallels many modern machine learning methods for prediction. These include reservoir computing \cite{tanaka2019recent}, backpropagation through unrolling time in recurrent neural networks \cite{liao2018reviving}, and transformers \cite{vaswani2017attention} which underlie the marked effectiveness of modern large language models, such as ChatGPT.

To begin the development, the following section briefly introduces thermodynamic machine learning from basic principles. Section~\ref{sec:Constrained} describes memory-constrained work optimisation. We then characterize the performance of information engines learned through maximum work production by calculating the work production rate for both training and test data. This leads to the following main results:
\begin{enumerate}
     \setlength{\topsep}{-5pt}
      \setlength{\itemsep}{-5pt}
      \setlength{\parsep}{-5pt}
\item We derive an exact expression for the asymptotic work rate of the information engine in Thm. \ref{thm:WorkRate}, relating the engine's estimated predictive model and the true predictive model of the input process.
\item We greatly simplify the search for the maximum-work engine by analytically deriving the engine parameters in Thm. \ref{Thm:edgeweights}.
\item We identify thermodynamic overfitting through divergent dissipated work in Fig. \ref{fig:TrainingAndValidation}.
\item We introduce thermodynamic regularizers that add a cost to model complexity during training. This results in thermodynamic learning that generalizes and avoids overfitting; see Figs. \ref{fig:SingleScatter} and \ref{fig:FiveStateComparison}.
\end{enumerate}
Altogether, these results demonstrate that thermodynamic principles spontaneously produce effective predictive learning.

\section{Thermodynamic Machine Learning}

Thermodynamic machine learning arises when a physical agent tries to extract energy from a complex, noisy environment. Facing such an environment the agent predicts the environment's behavior and then converts that knowledge into useful work. This section describes the salient aspects of this process. Here, the ``agent'' may be understood as a version of Maxwell's Demon~\cite{Maxw71, Thom74a, Leff2002} confronted with correlated patterns.

The work value of information has been thoroughly explored since Szilard's proposal of his eponymous information engine~\cite{Szilard1929} and the introduction of Landauer's erasure principle~\cite{Land61a}. In essence, given information-bearing degrees of freedom characterised by a random variable $Y$ we may, on average, extract an amount of work upper-bounded by $W_{ext}$  \cite{Parr15a}:
\begin{align}
    \beta W_{ext}= \Delta H \equiv H[Y']-H[Y]
    \label{eq:W_ext}
\end{align}
by converting $Y$ to an output random variable $Y'$ under coupling with an external heat bath at inverse temperature $\beta$. Here, $H[Y]=-\sum_{y \in \mathcal{Y}} \Pr(Y=y)\ln \Pr(Y=y)$ is the Shannon entropy of random variable $Y$.  For clarity, $\mathcal{Y}$ is the physical system, $y \in \mathcal{Y}$ are realizations of states from that system, and $Y$ is the random variable that determines the distribution over those system states via $\{\Pr(Y=y)\}_{y \in \mathcal{Y}}$.

We leave aside the ongoing investigation into the conditions for saturating this relationship~\cite{esposito2010entropy,Reeb2014,Taranto2023} and consider a procedure that saturates the bound. The maximal value for $\Delta H$ is attained by maximally randomizing $Y'$ such that $H[Y']=\ln |\mathcal{Y}|$, where $|\mathcal{Y}|$ is the number of distinct values that $Y$ may take.  An agent that fully extracts this maximal amount of work from an information source $Y$ is thermodynamically efficient, because it has harvested all available free energy.

We can further decompose the average work extracted by an efficient agent expressed in Eq.~(\ref{eq:W_ext}) into its single-shot elements:
\begin{align}
 \beta \avg{W(y)}=\ln \Pr(Y=y)+\ln|\mathcal{Y}|.
 \label{eq:W_y}
\end{align}
This quantifies the average extractable work from a probabilistically occurring realization $y$ in the maximum-extraction limit of full randomization.

This is the necessary single-shot work extraction for an \emph{efficient} agent, because it must have zero total entropy production $\langle \Sigma \rangle  \equiv \langle Q \rangle /T +k_B\Delta H =0$ on average and therefore must also have zero fluctuations in entropy production $ \Sigma(y)   \equiv  Q(y)  /T + k_B\ln \frac{p_y}{|\mathcal{Y}|}=0$ \cite{Croo99a,Boyd21a}.  Here, $Q/T$ is the entropy change in the environment due to heat $Q$ and $\Delta H$ is the change in entropy of the system $\mathcal{Y}$, which together account for the total entropy produced~$\Sigma$~\cite{esposito2011second,seifert2005entropy}.

In the case where the process transforms an information reservoir, the system starts and ends in an energetically degenerate configuration such that the work production and heat are equal with opposite signs. To minimize the entropy production when harvesting energy from $\mathcal{Y}$ through a quasistatic protocol, the probability of each outcome must be encoded in the initial energy landscape \cite{garner2017thermodynamics, Boyd21a}. The parameters of an agent's estimated input are explicitly encoded in its evolving energy landscape.

Here, we are here interested in the situation where $Y$ represents a stochastic process taking values $y_{0:L} \equiv y_0y_1 \cdots y_{L-1}$ generated by a model $\theta$ with probability:
\begin{align}
    \Pr(Y^\theta_{0:L}=y_{0:L}).
\end{align}  
That is, we consider sequences of length $L$, indicated by subscript $0:L$. The work extracted by an efficient agent is then a specific case of Eq.~(\ref{eq:W_y}) with $Y = Y^\theta_{0:L}$ \cite{Boyd21a, touzo2020optimal}:
\begin{align}
    \beta\langle  W^\theta(y_{0:L})\rangle = \ln \Pr(Y^\theta_{0:L}=y_{0:L})+L \ln |\mathcal{Y}|,
    \label{eq:TML}
\end{align}

However, there is an important conceptual addition. To extract work from $Y^\theta_{0:L}$, the agent must interact with each element in sequence. Thus, the agent cannot establish a simultaneous energy landscape over all elements. Instead, the engine requires a memory system $\mathcal{X}$ that tracks what has already been observed in the sequence such that it can make optimal estimates of the next input. A work extraction device that has access to such an internal memory is called an \emph{information engine}. Such a device is a type of information ratchet \cite{Mand012a,Boyd15a}, a physically instantiated type of stochastic Turing machine, also known as a ``Brownian computer'' \cite{strasberg2015thermodynamics}.  Figure~\ref{fig:InformationRatchet} provides an illustration.  Specifically, an information engine is a ratchet whose functionality is reduced to maximally randomizing the outputs while storing all relevant correlations in the machine memory.

\begin{figure}[htbp]
\centering
\includegraphics[width=\columnwidth]{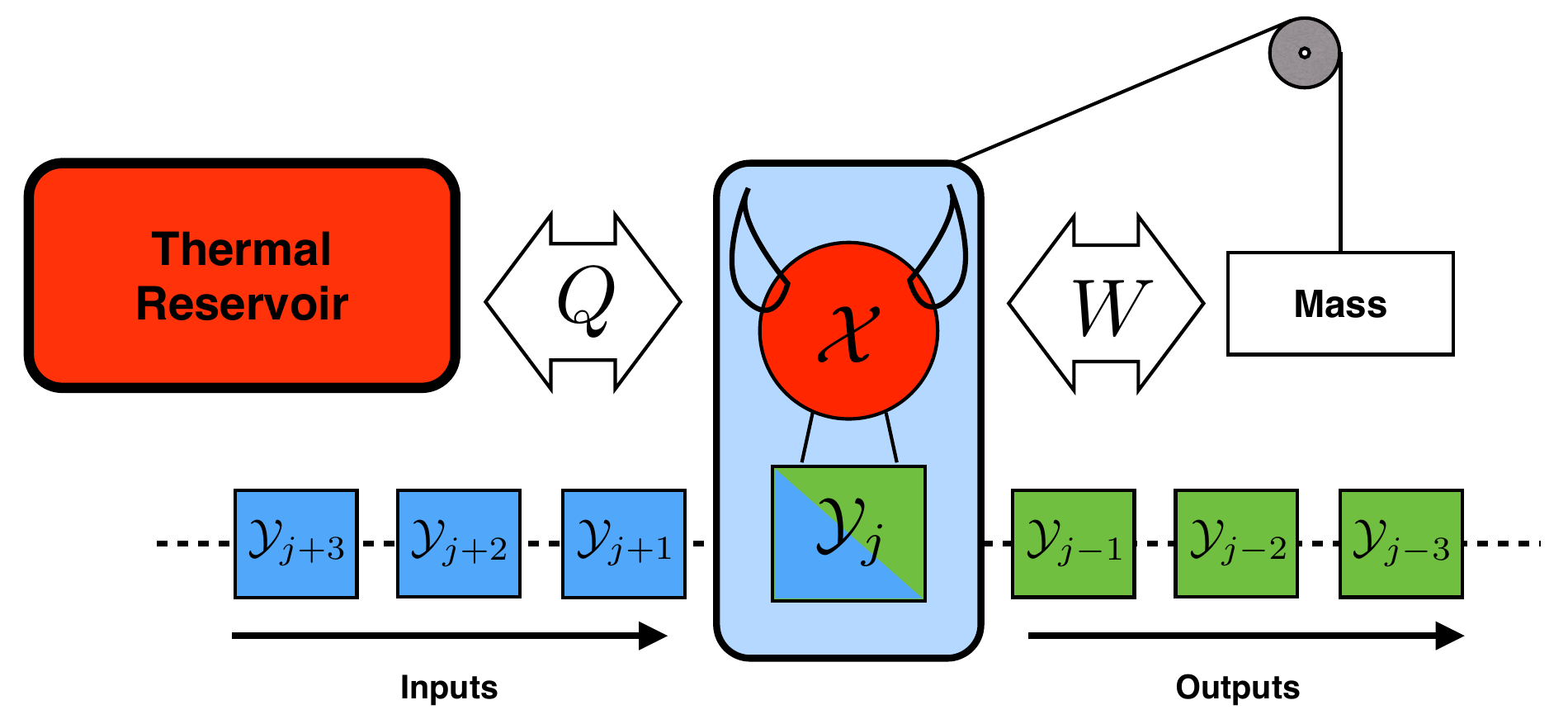}
\caption{An information engine: A physical device that can maximize work production using correlations on an information tape.  It transforms an input sequence $Y_{0:L}$ to an output sequence $Y'_{0:L}$ while exchanging energy between a work reservoir (represented by the mass on a string) and thermal reservoir at inverse temperature $\beta=1/\kB T$. The input $Y_j$ and output $Y'_j$ are stored in the physical system $\mathcal{Y}_j$.}
\label{fig:InformationRatchet}
\end{figure}

Faced with the output of a stochastic process $\mathcal{P}$, a thermodynamic agent must find a good process model $\theta$ among a candidate family of models $\Theta$. Denoting a given model $\theta$'s output distribution as $Y^\theta_{0:L}$, the likelihood of $\theta$ being the generator of a sequence~$y_{0:L}$ is:
\begin{align}
 \ell(\theta |y_{0:L}) \equiv \Pr(Y_{0:L}^\theta=y_{0:L}) ~.
\end{align}
To learn a model of he process $\mathcal{P}$ the agent may then employ \emph{maximum-likelihood estimation} (MLE) to select the model $\theta$ that maximizes the likelihood:
\begin{align}
    \Theta^\text{MLE}(y_{0:L}) = \underset{\theta \in \Theta}{\text{argmax}} \, \ell(y_{0:L}|\theta).
\end{align}
MLE is one of the most general techniques in machine learning. The resulting estimated distribution can be used for a variety of other learning tasks, including classification \cite{lin2017does}.

Next, let's consider in more detail how the agent models the process from which work is to be extracted. Note that generally any such process can be described by a \emph{hidden Markov model}~(HMM). Likewise, as the agent observes the process it must build its own model of the process. To minimize dissipation the information engine's memory must use the predictive states of the input process \cite{Boyd17a}. The agent's memory dynamics and energy landscape directly match a predictive model of $Y^\theta_{0:L}$ \cite{Boyd21a}.  The agent's task is thus to match its internal model to the process' true model.

A HMM is defined by hidden states $s$ and transitions between them according to the probabilities:
\begin{align}
     T^{(y)}_{s \rightarrow s'} \equiv \Pr(Y_i=y,S_{i+1}=s'|S_i=s),
\end{align}
outputting a symbol $y$ with the transition. Here, $Y_i$ and $S_i$ label the random variables corresponding to the output symbol and hidden state, respectively. $y$, $s$, $s'$ denote their realizations.

There are many ways to predictively model a time series $Y_{0:L}$. Among different procedures we choose minimal, predictive HMMs called \emph{$\epsilon$-machines}; see App. \ref{app:Review of Thermodynamic Machine Learning} along with examples there. For a given process there is no alternative predictive HMM that requires fewer hidden states and $\epsilon$-machines are sufficient for describing \emph{any} stochastic process giving rise to it~\cite{Crut01a}, even if nonstationary \cite{upper1997theory, Boyd21a}.  In addition, they provide a prescription for designing an information engine that harvests all available free energy from that process \cite{Boyd17a}. Thus, we choose our set of candidate models $\Theta \equiv \{ \theta \}$ to be a subset of the class of $\epsilon$-machines.  

\begin{figure}[tbp]
\centering
\includegraphics[width=1\columnwidth]{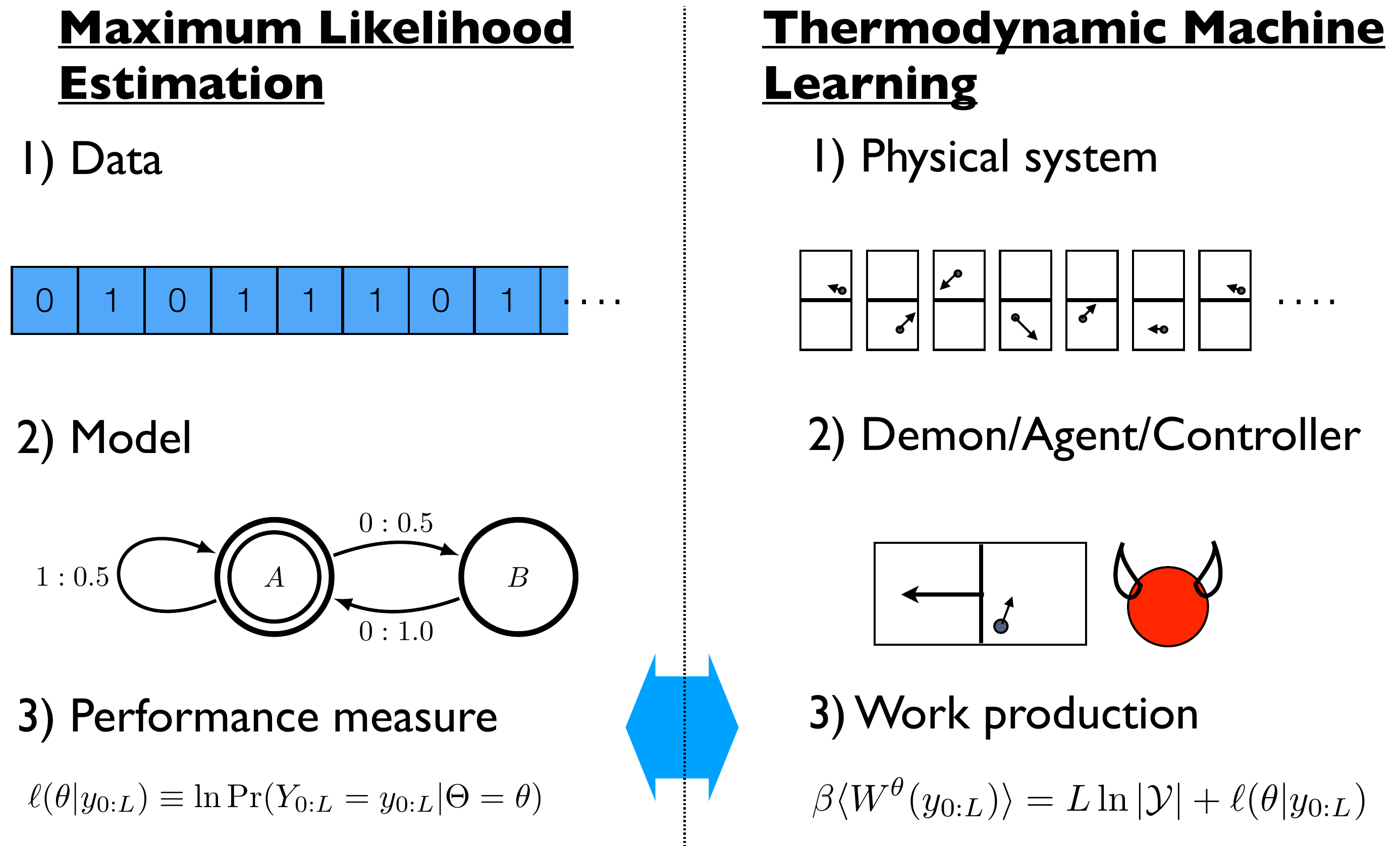}
\caption{Equivalence of thermodynamic machine learning and maximum likelihood estimation over predictive models. 1) MLE: Data (input information) is a realization of a random variable, such as a sequence of bits. TML: information is realized within a physical system, such as particles in partitioned boxes, each of which encodes a bit.  2) MLE: A model of the input information specifies the estimated probability of realizing data.  TML: An efficient demon, agent, or controller of the physical system that contains the input information must have an internal model of its estimated input.  3) MLE: The performance of the model for a particular input data is given by the log-likelihood.  TML: The performance of the efficient agent is measured by work production. The latter is linearly related to the log-likelihood.  MLE and TML are equivalent estimation processes.}
\label{fig:MLEvsTML}
\end{figure}

Finding the engine that produces the most work from given data is then equivalent to finding the engine whose model produces that information with maximum likelihood.  Thus, as illustrated in Fig.~\ref{fig:MLEvsTML} thermodynamic machine learning is equivalent to maximum likelihood estimation over predictive models \cite{Boyd21a}.  The input system encodes data, the information engine contains a model, and the performance of that engine (work production) scales proportionally to the log-likelihood. As illustrated in Fig.~\ref{fig:ThermodynamicTraining}, when performing thermodynamic learning on a collection of engines with models $\Theta$, we denote the maximum-work model for input $y_{0:L}$:
\begin{align}
    \Theta^\text{max}(y_{0:L})\equiv \underset{\theta \in \Theta  }{\text{argmax}} \langle  W^\theta(y_{0:L}) \rangle .
\end{align}
The equivalence between TML and MLE means that the inferred models from both learning strategies are the same:
\begin{align}
    \Theta^\text{max}(y_{0:L})=\Theta^\text{MLE}(y_{0:L}).
\end{align}

\begin{figure}[tbp]
\centering
\includegraphics[width=.8\columnwidth]{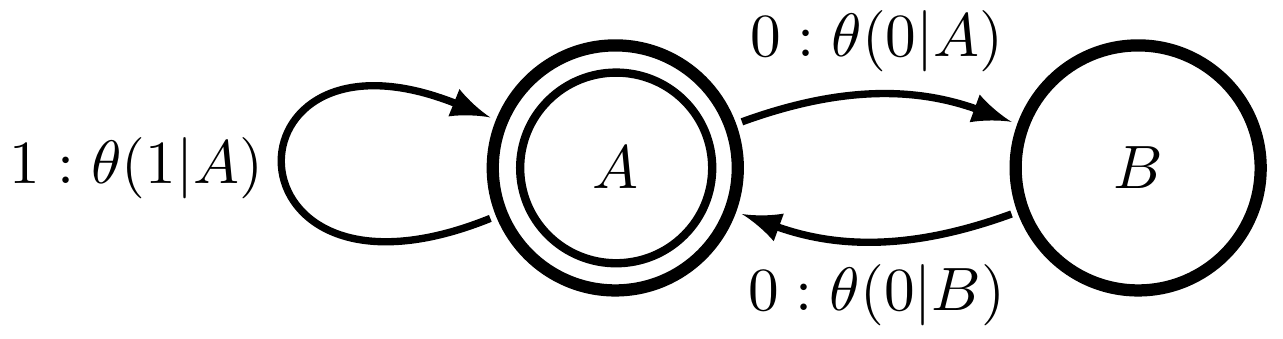}
\caption{Example $\epsilon$-machine: The Even Process with start state $A$ is described by $\theta$, which is composed of the causal states $\mathcal{S}$, outputs $\mathcal{Y}$, start state $s^*$, topology $S_{i+1}=\epsilon(Y_i,S_i)$, and edge-weights $\theta(Y_i|S_i)$. The Even Process produces sequences of zeros in even numbers.  The hidden states evolve according to $\epsilon(0,A)=B$, $\epsilon(1,A)=A$, $\epsilon(0,B)=A$, and transitions are taken with probabilities given by the edges weights $\theta(1|A)=0.5$, $\theta(0|A)=0.5$, and $\theta(0|B)=1.0$. Outputting a $0$ from $B$ has zero probability $\theta(1|B)=0.0$, so we leave $\epsilon(0,B)$ undefined. In this case, we chose the start state $s^*=A$.  Altogether, we can describe this model with the set  $\theta=\{\mathcal{S},\mathcal{Y},s^{*},\{\theta(y|s)\}_{y \in \mathcal{Y}, s \in \mathcal{S}},\{\epsilon(y,s)\}_{y \in \mathcal{Y}, s \in \mathcal{S}}\}$.}
\label{fig:ModelParameters}
\end{figure}

\begin{figure*}[htbp]
\centering
\includegraphics[width=1.7\columnwidth]{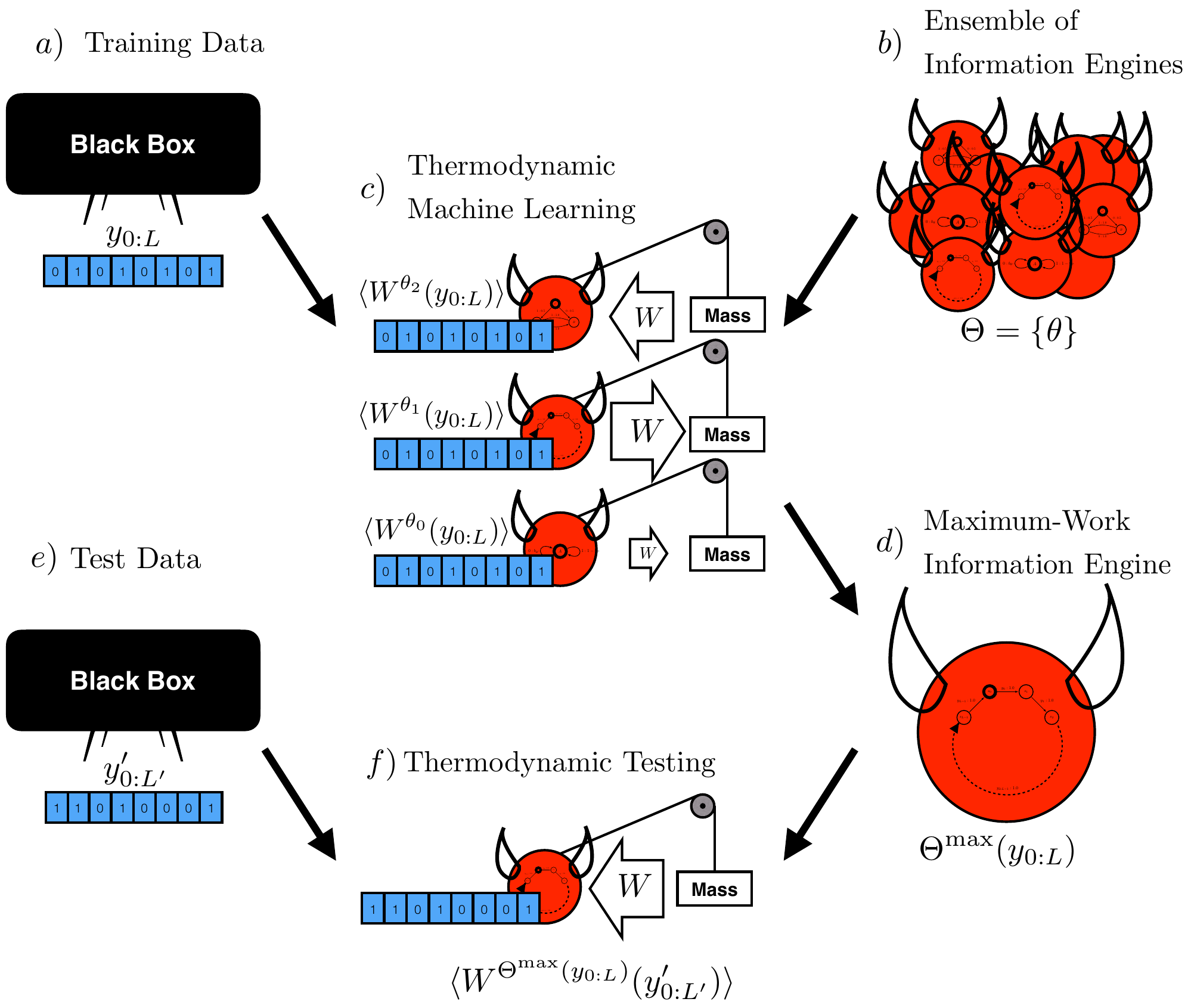}
\caption{Thermodynamic machine learning discovers the maximum-work information engine. This is followed by thermodynamic testing to validate its capacity to harvest energy using previously-unseen data:  a) We start with training data produced from a black box.  b)  Each candidate information engine from our ensemble has an internal model, represented faintly by the $\epsilon$-machine within their red body.  c)  We search through the engines to find that with the best model by determining how much work each produces from the training data. d) Thermodynamic machine learning converges on the maximum-work engine.  e) We take further test data from the black box of our environment. f) We feed the test data into our selected engine and track the work production to evaluate the effectiveness of the engine's model at capturing the process generated by the black box.}
\label{fig:ThermodynamicTraining}
\end{figure*}

Finally, $\epsilon$-machines can model any process and yield the same work production as any predictive model, meaning that it is sufficient to limit our model class $\Theta$ to them.  As described in App. \ref{app:Review of Thermodynamic Machine Learning} and Fig. \ref{fig:ModelParameters}, such machines are described by a causal update map on the hidden states (also called \emph{predictive states}):
\begin{align}
    S_{i+1}=\epsilon(Y_i,S_i),
\end{align}
and edge-weights:
\begin{align}
\theta(y|s) = \Pr(Y_i=y|S_i=s) .
\end{align}
These are the explicit parameters we must explore through training \cite{strelioff2014bayesian}.  The memory of the engine $\mathcal{X}$ is a direct copy of its model's predictive state space $\mathcal{S}$.   Once the maximum work model is found, we evaluate its complexity as the size of the engine's memory: $C=\ln|\mathcal{S}|=\ln|\mathcal{X}|$.

\section{Constrained Memory Work Maximization}
\label{sec:Constrained}

Starting from the equivalence between work maximization and MLE, the following details an algorithm for discovering predictive models of training data $y_{0:L}$ via TML.  The class of $\epsilon$-machines is a particularly appropriate model class for learning as they can produce any process $\Pr(Y_{0:\infty})$ given sufficiently many memory states (causal states) through the causal equivalence relation~\cite{Crut88a,Crut12a}. Since $\epsilon$-machines are the most general model class, potentially any pattern is discoverable via TML.  

However, the extreme generality of $\epsilon$-machines comes with a downside in learning.  For any sequence $y_{0:L}$ $\epsilon$-machines include the process that produces it with unit probability.  This means that if we allow our engine arbitrarily large memory $n=|\mathcal{X}|$, we can trivially maximize work production $ W^\theta(y_{0:L})$.  This corresponds to simply storing the training data in the engine's memory, rather than trying to discover the underlying pattern.  In this case, any other word besides the training word would be expected with zero likelihood, making this the extreme limit of overfitting.  This limiting case demonstrates that there are learning algorithms for which double-descent \cite{nakkiran2021deep} does not apply.  Moreover, it is natural to constrain memory, because it is an informational resource.

We consider the maximum-work model from the set $\Theta_n$ of $\epsilon$-machines with $n$ predictive states. This means that the engine is limited to $n$ memory states:
\begin{align}
    \Theta_n^\text{max}(y_{0:L})\equiv \underset{\theta \in \Theta_n}{\text{argmax}} \langle  W^\theta(y_{0:L}) \rangle.
\end{align}
This engine produces work during training equivalent to:
\begin{align}
   \langle W_n^\text{max}(y_{0:L}) \rangle \equiv \underset{\theta \in \Theta_n}{\text{max}} \langle  W^\theta(y_{0:L}) \rangle.
\end{align}
This is the training we execute throughout our development here: Find the maximum-work $n$-state engine that corresponds to the maximum-likelihood model from the class of predictive $n$-state HMMs.  Unlike other training algorithms---that rely on the convergence of numerical estimators---this training is essentially analytic. Given our chosen class of models, as long as we successfully enumerate all accessible topologies, we directly find the maximum-work engine among the candidates.

Recall that the outcomes of many machine learning algorithms change based on how the learning parameters are selected and whether local maxima can be escaped. In contrast, the TML algorithm always arrives at the same model for the same input data and definition of work production. The following explains how this is implemented.

The challenge of finding the maximum-work model from such a general class of processes may at first seem like a daunting optimization process over the high dimensional space of symbol-labeled stochastic matrices that specify an $\epsilon$-machine's Hidden Markov Model (HMM).  For general classical and quantum HMMs, evaluating the likelihood of $y_{0:L}$ requires a series of $L$ linear operations \cite{adhikary2020expressiveness}.  However, the properties of unifilarity allow for a useful simplified expression for the work production in terms of the $\epsilon$-map and edge-weights, as shown in Appendix \ref{app:Maximum Work Edge Weights}:
\begin{align}
\beta \langle W^\theta(y_{0:L}) \rangle = L \ln |\mathcal{Y}| + \ln \prod_{i=0}^{L-1} \theta(y_i|\epsilon(y_{0:i},s^*)).
\end{align}
For a particular topology $\epsilon$, start state $s^*$, and input word $y_{0:L}$, we can analytically find the maximum-work edge-weights by counting the number of times that the input $y$ is received by predictive state $s$ due to the input-driven dynamics of the predictive state:
\begin{align}
    N(y,s|s^*,\epsilon,y_{0:L}) \equiv \sum_{i=0}^{L-1} \delta_{y,y_i}\delta_{s,\epsilon(y_{0:i},s^*)}.
\end{align}

As shown in App. \ref{app:Maximum Work Edge Weights}, the work production can be rewritten:
\begin{align}
\beta \langle W^\theta(y_{0:L}) \rangle & = L \ln |\mathcal{Y}|
\\ & + \sum_{s,y} N(y,s|s^*,\epsilon,y_{0:L})\ln \theta(y|s). \nonumber
\end{align}
The resulting maximum-work edge-weights are derived using the method of Lagrange multipliers. They are simply the fraction of times that predictive state $s$ receives $y$ when driven by $y_{0:L}$:
\begin{align}
    \Theta^\text{max}_{s^*,\epsilon,y_{0:L}}(y|s)=\frac{N(y,s|s^*,\epsilon,y_{0:L})}{\sum_{y'} N(y',s|s^*,\epsilon,y_{0:L})}.
\end{align}
This is purely a frequentist estimate of the edge-weights based on how often each input visits each predictive causal state. This thermodynamically motivated engine design directly reflects the results of training a transformer, as described in Ref. \cite{basu2023transformers}. 

\begin{figure}[tbp]
\centering
\includegraphics[width=\columnwidth]{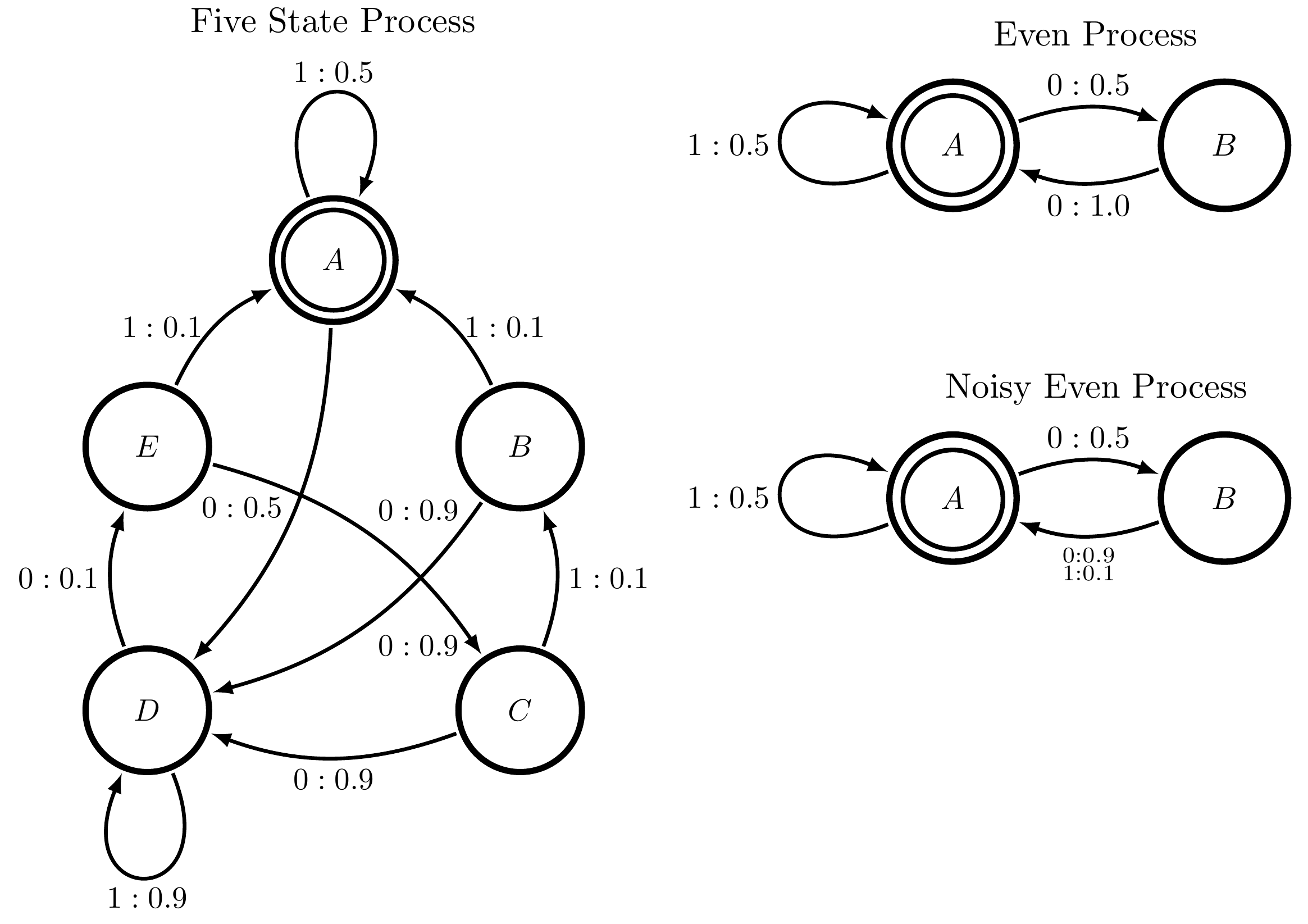}
\caption{Example $\epsilon$-machines: the ``Five-State Process,'' which is a randomly generated process with five causal states and full support, the ``Even Process,'' which produces sequences of $0$s in even numbers, and the ``Noisy Even Process'' which adds some noise to the Even Process such that it has full support.}
\label{fig:ThreeMachines}
\end{figure}

\begin{figure*}[tbp]
\centering
\includegraphics[width=2\columnwidth]{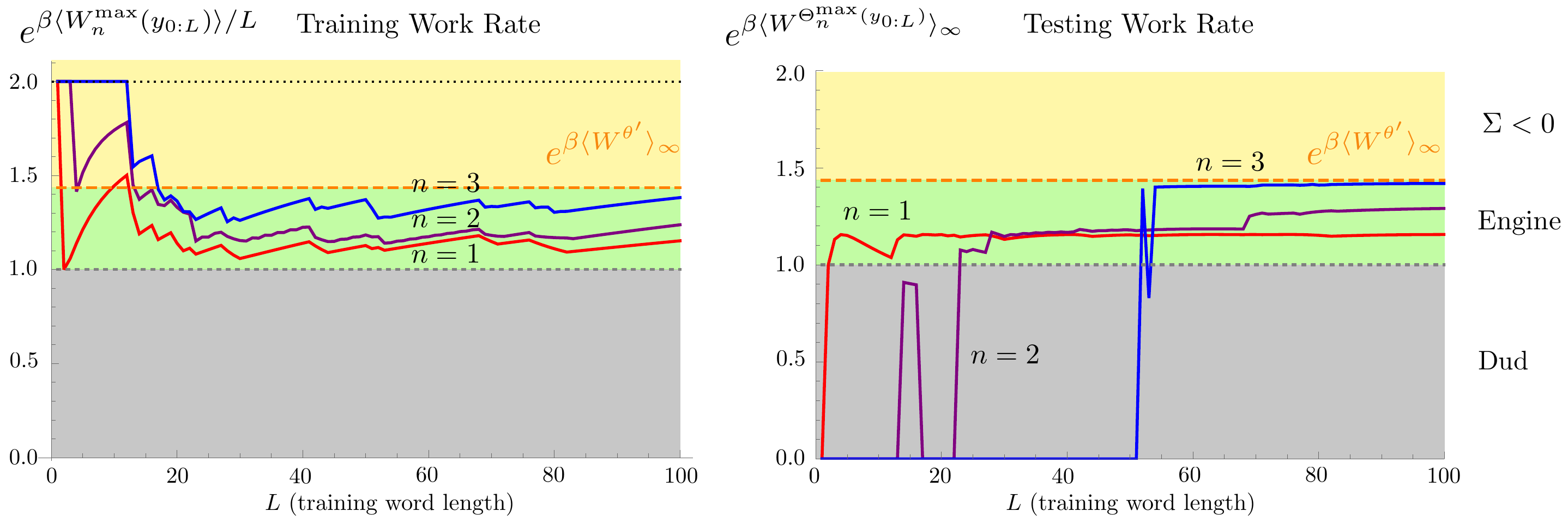}
\caption{Thermodynamic Overfitting: The work rate converges quickly on the training data (left), but this hides the divergent energy dissipation during testing (right) that corresponds to thermodynamic overfitting.  Work production from training data increases as the number of the engine's predictive states $n$ increases.  However, we see that asymptotic work production during testing is minimal ($-\infty$) for small amounts of training data, especially for more complex engines, which corresponds to overfitting.  For reference, we plot (i) the grey dashed line for zero work production $\beta \langle W^\theta \rangle_\infty=0$, (ii) the orange dashed line for work production that corresponds to correctly estimating the true model $\beta \langle W^{\theta'} \rangle_\infty$), and (iii) the black dotted line for $\beta \langle W^\theta (y_{0:L}) \rangle/L= \ln 2$ is the maximum work that can be harvested per bit from a single sequence. In the grey region below the grey dashed line, the engine is unable to harvest work on average, so it is a ``Dud''.  In the green region just above the grey dashed line, positive work is produced, effectively functioning as an ``Engine''.  Last, in the yellow region above the orange dashed line, we labeled the region with $ \Sigma  <0$ to indicate a negative entropy fluctuation.} 
\label{fig:TrainingAndValidation}
\end{figure*}

Returning to analytically finding the maximum-likelihood model and maximum-work engine with $n$ memory states, this reduces to checking every allowed topology \cite{strelioff2014bayesian}.  In Fig. \ref{fig:TrainingAndValidation} we do this exactly for $n \in \{1,2,3 \}$ for a word $y_{0:100}$, for which we train on all intermediate length strings $\{y_{0:L}\}_{L \in \{1,2, \cdots 99, 100\}}$.  The word was generated from the Five-State $\epsilon$-machine shown in Fig. \ref{fig:ThreeMachines}.  Since we are limited to models $\theta$ with three or fewer causal states and the underlying model $\theta'$ has five, we have technically misspecified our model class \cite{claeskens2008model}.

We consider the work production rate of the maximum work model $\beta \langle W^\text{max}(y_{0:L}) \rangle/L$ in Fig. \ref{fig:TrainingAndValidation}, since it determines how much energy each symbol contributes on average during training.  Taking the limit $L \rightarrow \infty$, this should approach $\langle W^\theta \rangle_\infty$---the average work extraction per bit for the engine, where $\theta$ is the engine's model.

While the next section derives an explicit expression for $\langle W^\theta \rangle_\infty$, for now we discuss the length-$L$ work production rate of the maximum-work engine. This is plotted with respect to three regions of asymptotic functionality:
\begin{enumerate}
\item Dud: An information engine that has nonpositive work production rate is a Dud, because it produces no useful energy.
\item Engine: An information engine that produces positive work is a functional engine.  
\item Negative entropy production ($ \Sigma   <0$):  $\Sigma$ denotes the total entropy production in the thermal environment and information bearing degrees of freedom combined \cite{seifert2005entropy}.  If the engine produces more work per bit than an engine that estimates a true model $\langle W^{\theta'} \rangle_\infty$, then it has extracted more energy than is available in the form of free energy per symbol. If it continues at this work production rate, the total entropy production will be negative on average and violate the Second Law of thermodynamics.
\end{enumerate}
We also plot a dashed black line at $\beta \langle W^\theta \rangle = \ln 2$, which is the Landauer benefit of randomizing a bit \cite{Land61a}.

Figure \ref{fig:TrainingAndValidation} shows that all the maximum-work information engines produce positive work and do not exceed Landauer's bound.  This is because it is always possible to find an engine that gains energy by interacting with a singular realization of the true model (the data).  We also see that early in training (for short length strings $L< \sim 20$) the training work rate often achieves the $k_BT \ln 2$ Landauer limit on work production.  This corresponds to discovering a predictive model that deterministically produces the training sequence with unit probability \cite{vieira2022temporal}.  However, this is well into the regime of asymptotic Second Law violation if the true source is anything other than a deterministic repetition of the training string, meaning that entropy production for these instances is negative.  This is possible since the data is a single realization: negative entropy fluctuations are allowed as long as the detailed fluctuation theorem is satisfied and entropy is nondecreasing \emph{on average} \cite{Croo99a}.

As the training length increases in Fig. \ref{fig:TrainingAndValidation}, a general trend appears: The training work-rate slowly ramps up, interspersed by sudden sharp dips.  These dips happen at roughly the same point for all three curves ($n=1$, $n=2$, and $n=3$).  
The slow ramps of increasing work rate correspond to the slow accumulation of low-suprisal symbols, where we can evaluate the surprisal of the next symbol $-\ln \Pr(Y^\theta_L=y_L|Y^\theta_{0:L}=y_{0:L})$ via the probability conditioned on the past sequence. The ensuing dip corresponds to breaking the sequence with a rare high-surprisal element, significantly increasing the estimated input entropy rate and reducing the work rate.  

Figure \ref{fig:TrainingAndValidation} also shows that memory size significantly affects work production. There is a consistent thermodynamic advantage to additional memory in training.  Work production increases as memory increases $\langle W_{n+1}^\text{max}(y_{0:L}) \rangle \geq \langle W_n^\text{max}(y_{0:L}) \rangle$.  This is as expected, because every $n$-state model can be described using $n$ or more states.

The thermodynamic advantage of higher memory parallels the thermodynamic Principle of Requisite Complexity \cite{Boyd17a}---an engine's memory must match the predictive complexity of its fuel to optimally leverage all correlations.  The challenge of harvesting energy from an information source $\theta'$ closely parallels the challenge of predicting that same source, and more memory yields better prediction of temporally correlated information  \cite{jurgens2021shannon}.  The advantage of larger memory is especially pronounced for short training lengths, where we see that 3-state machines can produce work at the limit of the Second Law for far longer training lengths.

One might infer from the advantage of additional memory in training that more memory is always thermodynamically advantageous.  However, we will now explore a more nuanced picture of the costs and benefits of engine complexity by analyzing the effectiveness of the engine harvesting energy from further inputs from the information source.

\section{Asymptotic Work Harvesting and Overfitting}

After training to identify the maximum-work information engine, we examine its thermodynamic performance on test data as shown in Fig. \ref{fig:ThermodynamicTraining}. The purpose of training is to find an engine whose internal model approximates the true model $\theta'$ that produced the training data.  TML has found a ``good model'' if the resulting engine is able to effectively produce work when it receives further inputs from the true source.  Thus, we consider the average work that would be produced in the testing stage.

If the $\epsilon$-machine of the true input process is described by:
\begin{align*}
    \theta'=\{\mathcal{S}',\mathcal{Y},s^{*'},\{\theta'(y|s')\}_{y \in \mathcal{Y}, s' \in \mathcal{S}'},\{\epsilon'(y,s')\}_{y \in \mathcal{Y}, s' \in \mathcal{S}'}\},
\end{align*}
then an efficient information engine with internal model:
\begin{align*}
    \theta=\{\mathcal{S},\mathcal{Y},s^{*},\{\theta(y|s)\}_{y \in \mathcal{Y}, s \in \mathcal{S}},\{\epsilon(y,s)\}_{y \in \mathcal{Y}, s \in \mathcal{S}}\},
\end{align*}
will on average produce work over $L$ time steps:
\begin{align}
\begin{split}
    \beta \langle W^\theta \rangle_{0:L} =& \sum_{y_{0:L}} \Pr(Y^{\theta'}=y_{0:L}) \ln \Pr(Y^{\theta}=y_{0:L})\\&+L \ln |\mathcal{Y}|.
    \end{split}
\end{align}
$Y^{\theta'}_{0:L}$ is the random variable of the actual input process over length $L$.  The average work produced during the $L$th time step is:
\begin{align}
    \beta \langle W^\theta \rangle_{L} \equiv \beta \langle W^\theta \rangle_{0:L+1}-\beta \langle W^\theta \rangle_{0:L}.
\end{align}
We consider the asymptotic rate of work production:
\begin{align}
\langle W^\theta \rangle_\infty \equiv \lim_{L \rightarrow \infty}  \langle W^\theta \rangle_{L},
\end{align}
as the measure of an engine's effectiveness.

\begin{The}
\label{thm:WorkRate}
The asymptotic work rate for an efficient engine with internal predictive model $\theta$ when harvesting information from a source with predictive model $\theta'$ is:
\begin{align}
\beta \langle W^\theta \rangle_\infty =\ln |\mathcal{Y}| + \sum_{s,s',y} \pi_{s,s'} \theta'(y|s')  \ln \theta(y|s).
\end{align}
Here, $\pi_{s,s'}$ is the steady-state of the joint hidden states $\mathcal{S} \otimes \mathcal{S}'$ if the causal update of $\theta$ is driven by $\theta'$:
\begin{align}
\pi_{s_1,s_1'} = \sum_{s_0,s_0',y} \delta_{s_1,\epsilon(s_0,y)}  \delta_{s'_1,\epsilon'(s'_0,y)}  \theta' (y| s'_0) \pi_{s_0,s'_0}.
\end{align}
\end{The}

\begin{proof}
See Appendix \ref{app:Derivation of Asymptotic Work Rate}.
\end{proof}

Using the expression in Thm. \ref{thm:WorkRate}, we can calculate the rate of entropy production (dissipated work) \cite{Parr15a} by comparing the work rate to the rate of nonequilibrium free energy change:
\begin{align}
    \langle \Sigma^\theta \rangle_\infty/k_B =\beta (-\langle W^\theta \rangle_\infty-\Delta F^\text{NEQ}_\infty).
\end{align}
This is the amount of free energy wasted per symbol, which quantifies the irreversibility of the information engine.  As shown in App. \ref{app: Entropy Production as Divergence}, the entropy production can be reduced to the average relative entropy between the next-input prediction of the true model's $\theta'$ hidden state $s'$ and the prediction of the estimated model's $\theta$ hidden state $s$:
\begin{align}
    \langle \Sigma^\theta \rangle_\infty/k_B  = \sum_{s,s'} \pi_{s,s'}  D_{KL} (Y^{\theta'}_i|S'_i=s'||Y^{\theta}_i|S_i=s),
\end{align}
where:
\begin{align*}
D_{KL}(X'|Z'=z'||X|Z=z) \equiv & \\
\sum_{x}\Pr(X'=x|Z=z') & \ln \frac{\Pr(X'=x|Z'=z')}{\Pr(X=x|Z=z)}
\end{align*}
denotes the relative entropy between the distribution on $X'$ induced by the condition that $Z'$ realizes $z'$ and the distribution on $X$ induced by the condition that $Z$ realizes the element $z$. Such relative entropies appear as the additional dissipation incurred by misestimating the input distribution \cite{Kolc17a, Riec20a}.  If a learning process refines and improves the estimator $\theta$, its divergence from the actual process $Y^{\theta'}_{0:L}$ should diminish, monitoring how much learning reduces entropy production \cite{milburn2023quantum, gold2019self}.

Figure \ref{fig:TrainingAndValidation} plots the asymptotic work rate for the maximum-work models $\{\Theta^\text{max}_n(y_{0:L})\}_{L \in \{1,2,\cdots ,99, 100 \}}$ that result from training on the words $\{y_{0:L}\}_{L \in \{1,2,\cdots ,99, 100 \}}$ appearing in $y_{0:100}$ which was randomly generated by the Five-State model shown in Fig. \ref{fig:ThreeMachines}. The result is compared to the engine's maximum possible work rate with the true model $\langle W^{\theta'} \rangle_\infty$. The difference between these work rates gives the entropy production rate.

Again, we decompose the regions of functionality into Dud (negative work production), Engine (positive work production), and $\Sigma < 0$ (Second Law violation).  We see that all learned engines respect the Second Law and extract less work than if they had estimated the true model $\theta'$.  The upper bound for engines is the work rate that results from guessing the true model, as this is the difference in entropy rates between inputs and outputs. According to the Information Processing Second Law of thermodynamics, this entropy difference is the accessible free energy per symbol and bounds the work production \cite{Boyd15a}.

Figure \ref{fig:TrainingAndValidation} also shows a thermodynamic advantage in the testing work rate for engines with larger memories, if trained on long words.  While it is unclear how close to the optimal $n$-state engine these results are, we see that thermodynamic learning discovers enough of the hidden temporal structure to harvest much of the available free energy. In fact, the three-state information engine nearly achieves the thermodynamic limit of perfect efficiency for training length $L= \sim 100$.  Rate-distortion theory \cite{Stil07b, creutzig2009past, marzen2016predictive} provides a prescription for the most predictive model given a limitation on available memory. 

This noted, training engines with additional memory produces less effective engines for short training words.  In fact, the maximum-work three-state engine produces $-\infty$ work in testing for training lengths up to $L \approx 50$.  Divergent dissipated work in Fig. \ref{fig:TrainingAndValidation} corresponds to \emph{overfitting}, where the number of available model parameters is much larger than can be reasonably deduced from frequentist estimates from the data.  Each memory state must make a prediction of its input. For short training words and larger memories, it becomes more probable that one of the memory states will not receive any copies of one of the input symbols $N(y,s|s^*,\epsilon,y_{0:L})=0$.  In this case, the engine estimates $\theta(y|s)=0$, which comes at the cost of negative divergent work production if such an observation is actually possible for the input process $\theta'$ \cite{Kolc17a, Riec20a}.  The size of the available parameter space within even $3$-state models is large enough that it realizes the main feature of overfitting for this particular training word---good performance on training data, while failing to effectively predict and harvest energy from test data.  There is an asymptotic benefit to memory, but there are also dire costs to using an excessively complex engine when training data does not justify it.

In this way, we identified thermodynamic overfitting,  showing that it is a considerable hurdle for TML. We now turn to resolve this challenge.

\section{Thermodynamic Generalization and Regularization}

Overfitting is commonly encountered in machine learning, as models with many degrees of freedom can encode a dataset explicitly in model parameters  \cite{hastie2009elements}.  To circumvent it, it is commonplace to implement regularization techniques that allow models to generalize and better predict unseen data. These strategies include restrictions that limit the model complexity that can be discovered through training. See, for example, ``dropout'' in deep neural networks \cite{srivastava2014dropout} and regularizers that add an explicit penalty to model complexity to modify the performance measure \cite{tibshirani1996regression, zhang2010regularization}. 

The following considers two physically-motivated methods of \emph{thermodynamic regularization}:
\begin{enumerate}
    \item \emph{Autocorrection}: An engine cannot start synchronized with the true predictive state of the input. Through the influence of the input symbols on the engine's predictive state dynamics, the engine must autocorrect in order to synchronize its estimated predictive state. Not surprisingly, there is a transient energy cost as the engine autocorrects its predictive state and approaches its steady-state dynamics \cite{Boyd16c}.
    \item \emph{Engine Initialization}: There is an energy cost to initializing the energy landscape of a predictive information engine.
\end{enumerate}
We will now show that autocorrection incurs an additional cost to complex models with unnecessarily distinct predictions from each memory state. In addition, we find that the cost of engine initialization leads to Bayesian updates of the edge-weights according to Laplace's rule of succession \cite{jaynes2003probability}.

Let $C(\theta)$ denote the energy penalty of initializing an engine with model $\theta$. It functions as a \emph{regularizer}.  If we also include the cost of autocorrecting from an initial distribution $p(s_0)\equiv \Pr(S_0=s_0)$, as shown in App. \ref{app:Maximum Work Edge Weights}, the average work production can be expressed:
\begin{align*}
    \beta \langle W^\theta_G (y_{0:L}) \rangle & = L \ln |\mathcal{Y}| -C(\theta) 
    \\ &+ \sum_{s_0}p(s_0) \ln \Pr(Y^\theta_{0:L}=y_{0:L}|S_0=s_0) ,
\end{align*}
where $\Pr(Y^\theta_{0:L}=y_{0:L}|S_0=s_0)$ is the probability of model $\theta$ producing $y_{0:L}$ when starting from predictive state $s_0$.  As with the case of unregularized TML, we can re-express the work production by tracking the state dynamics induced within the predictive states by the input word.  Again, we need only count the number of times $N(y,s|s_0,\epsilon,y_{0:L})$ that $y$ is input to predictive state $s$ for every the initial state $s_0$, then we obtain the work production:
\begin{align}
    \beta \langle W^\theta_G (y_{0:L}) \rangle & = L \ln |\mathcal{Y}| -C(\theta) 
    \\ &+ \sum_{s_0,s,y}p(s_0)  N(y,s|s_0,\epsilon,y_{0:L})\ln \theta(y|s). \nonumber
\end{align}

The penalty $C(\theta)$ for model $\theta$ incurred through training is motivated by physical constraints, such as the energetic cost of initializing the model.  If TML evaluates a model that is too complex and costly to initialize, then that cost should counteract the energetic benefit of predicting the training word.  This echoes the \emph{minimum description length principle}, in which the benefit gained through prediction is offset by the cost of describing the model \cite{grunwald2007minimum}. In this spirit, we introduce an energy penalty of preparing the edge-weights $\theta(y|s)$ of a particular model $\theta$.

The penalty originates from the fact that equilibrium probabilities $\pi$ are directly related to energies via:
\begin{align*}
\beta E(z) = \beta F^\text{eq}-\ln \pi(z).
\end{align*}
The equilibrium free energy $-\beta F^\text{eq} \equiv \ln \sum_{z} e^{-\beta E(z)}$ is the upper limit on work that can be produced from an equilibrium distribution.  The equilibrium distribution is the necessary starting point for an efficient quasistatic transformation of information \cite{Jun14a, garner2017thermodynamics, Boyd17a}.  This means that if we prepare a distribution $q(z)$ that differs from the initial equilibrium distribution, we can calculate the dissipated work $\langle W_\text{diss} \rangle$ as it partially relaxes to $q'(z)$ via a difference in relative entropies \cite{Parr15a, Riec20a, Kolc17a}:
\begin{align*}
    \beta \langle W_\text{diss} \rangle =D_{KL}(q||\pi)-D_{KL}(q'|\pi).
\end{align*}
The dissipation associated with preparing every edge-weight of the model $\theta$ should contribute to a Thermodynamic Regularizer.

As discussed in App. \ref{app:Maximum Work Edge Weights}, preparing the edge-weight $\theta(y|s)$ of a particular combination of predictive state $s$ and input $y$ incurs the dissipated work of:
\begin{align*}
    \beta \langle W^\text{prepare}_\text{diss}(s,y) \rangle=-\ln \theta(y|s).
\end{align*}
We propose that the cost of implementing a model is proportional to the dissipation associated with preparing every edge-weight:
\begin{align*}
    C(\theta) & = \alpha \sum_{s,y} \beta \langle W^\text{prepare}_\text{diss} \rangle
    \\ & = -\alpha \sum_{s,y} \ln \theta(y|s),
\end{align*}
where $\alpha$ is a regularization parameter of our choosing.  For the parameter $\alpha=1$, this is the total energetic excess beyond the free energy for each combination engine memory state and input:
\begin{align*}
    C(\theta)= \sum_{s,y}\beta(E^\theta(s,y)-F^\theta(s)).
\end{align*}
We associate this additional cost with examining the relaxation to check the estimated probability of every edge-weight. This is a \emph{thermodynamic regularizer} in that the penalty originates from a thermodynamic implementation of the engine with model $\theta$.

The resulting regularized work production is
\begin{align}
    \beta \langle W^\theta_{p,\alpha,\epsilon} (y_{0:L}) \rangle & = L \ln |\mathcal{Y}| -\alpha  \langle W^\text{prepare}_\text{diss} \rangle
    \\ &+ \sum_{s_0}p(s_0)  N(y,s|s_0,\epsilon,y_{0:L})\ln \theta(y|s). \nonumber
\end{align}
This work production can be maximized analytically.

\begin{The}
\label{Thm:edgeweights}
The maximum-work edge-weights of an engine with input $y_{0:L}$, topology $\epsilon$, and regularization parameters $\alpha$ and $p$ can be analytically calculated by tracking the dynamics of the $\epsilon$-map from each start state:
\begin{align*}
\Theta^\text{max}_{p,\alpha,\epsilon,y_{0:L}}(y|s) & = \frac{\sum_{s_0} p(s_0)(\alpha+N(y,s|s_0,\epsilon,y_{0:L}))}{\sum_{y',s_0} p(s_0)(\alpha+N(y',s|s_0,\epsilon,y_{0:L}))}.
\end{align*}
\end{The}
\begin{proof}
See App. \ref{app:Maximum Work Edge Weights}
\end{proof}

Theorem \ref{Thm:edgeweights} gives a shortcut to regularized TML.  Once given the topology of a candidate information engine, we need only track the memory dynamics and directly calculate the maximum-work edge-weights from $\Theta^\text{max}_{p,\alpha,\epsilon,y_{0:L}}(y|s)$. Thus, given a memory constraint, we scan through the available engine topologies and select that which produces the most work.  While the set of topologies grows super-exponentially with memory, it is vastly simpler than discovering the edge-weights through numerical optimization of edge-weights.

\subsection{Autocorrection}

When harvesting temporal correlations from a time series, the engine may start out of sync with the input process. As a result, it must \emph{autocorrect} to the predictive states input. Autocorrection requires additional energy \cite{Boyd16c}, as does synchronizing with the inputs \cite{Boyd16e}. To address this, consider the average work production if the engine starts in distribution $p(s)\equiv \Pr(X^\theta_0=s)$ over its memory states. Appendix \ref{app:Work Production From Distributed Start State} shows that the resulting average work production is the weighted log-likelihood of each predictive state:
\begin{align*}
& \beta \langle W_\text{AC}^\theta(y_{0:L}) \rangle 
\\ & = L \ln |\mathcal{Y}| +\sum_{s} p(s) \ln \Pr(Y^\theta_{0:L}=y_{0:L}|S_0^\theta=s) \nonumber
\\ & = L \ln |\mathcal{Y}| +\sum_{s} p(s) \ln \prod_{i=0}^{L-1} \theta(y_i|\epsilon(y_{0:i},s)). \nonumber
\end{align*}

As an example, consider the case where the initial distribution over memory states is uniform $p(s)=1/|\mathcal{X}|$.  Rather than finding the $n$-state model that maximizes the work production $\Theta^\text{max}_n(y_{0:N})$, we find the model that produces maximum work when we take into account the cost of autocorrection:
\begin{align}
    \Theta^\text{AC}_n (y_{0:L}) \equiv \underset{\theta \in \Theta_n}{\text{argmax}} \langle  W_{AC}^\theta(y_{0:L}) \rangle.
\end{align}

This sets us up to introduce a thermodynamic complexity measure for autocorrection. Note that requiring the agent to autocorrect introduces unavoidable energy inefficiencies that are not present in the unregularized MLE/TML strategy.  Given an initial density over the input's predictive states $p(s_0)=\Pr(S'_0=s_0)$, we have the true input distribution:
\begin{align*}
\Pr(Y^{\theta'}_{0:L})=\sum_{s_0}p(s_0)\Pr(Y^{\theta'}_{0:L}|S_0=s_0).
\end{align*}
The average entropy production is the difference between the average work production and the change in free energy:
\begin{align*}
    \langle \Sigma^\theta_\text{AC} \rangle_{0:L}/k_B & = \beta (-\langle W^\theta_\text{AC} \rangle_{0:L}-\Delta F^\text{NEQ}_{0:L})  
    \\ & = \sum_{s_0} p(s_0) D_{KL}(Y^{\theta'}_{0:L}||Y^\theta_{0:L}|S_0=s_0). \nonumber
\end{align*}
The average dissipation is proportional to the average divergence between the actual input and the estimated input from each predictive state of the estimated model $\theta$.  

Even if the agent correctly guesses the underlying $\epsilon$-machine model, such that the state transitions $\epsilon'=\epsilon$ and the edge-weights $\theta'(y|s)=\theta(y|s)$ are the same and the prediction from each causal state is the same:
\begin{align}
\Pr(Y^{\theta'}_{0:L}|S'_0=s_0)=\Pr(Y^{\theta}_{0:L}|S_0=s_0),
\end{align}
the agent can still produce entropy due to the lack of model synchronisation.  

The entropy produced through autocorrection is a new model complexity measure.  Specifically, the complexity is proportional to the cost of autocorrecting to the correct predictive states of the process from an initially uniform distribution over all predictive states:
\begin{align*}
\mathcal{C}_\text{AC}(\theta,L) \equiv \frac{\sum_{s_0}  D_{KL}(Y^{\theta'}_{0:L}||Y^{\theta'}_{0:L}|S'_0=s_0)}{|\mathcal{S}|}.
\end{align*}
Previous explorations of correlation-powered information engines show that they must autocorrect to synchronize with their inputs, incurring a thermodynamic cost along the way \cite{Boyd16c}.

The autocorrection cost $\mathcal{C}_\text{AC}(\theta,L)$ differs from a measure of stored information like statistical complexity $C_\mu$ \cite{Crut01a}. On the one hand, we can have many different possible configurations with the latter, with nearly identical predictions, but the difference between predictions of each causal state is immaterial to the complexity measure.  On the other hand, the complexity measure $\mathcal{C}_\text{AC}(\theta,L)$ reflects the thermodynamic cost of choosing a model that anticipates wildly divergent futures from different predictive states.  

There are cases in which we know this will approach a finite quantity in the asymptotic limit $L \rightarrow \infty$. (For instance, when $Y^{\theta'}_{0:\infty}$ has finite Markov order.)  In such cases, we can refine the complexity measure to be the thermodynamic cost of exactly synchronizing to the input process:
\begin{align*}
\mathcal{C}_\text{AC}(\theta) & \equiv  \lim_{L \rightarrow \infty} \mathcal{C}_\text{AC}(\theta,L)
\\ & =\frac{\sum_{s_0}  D_{KL}(Y^{\theta'}_{0:\infty}||Y^{\theta'}_{0:\infty}|S'_0=s_0)}{|\mathcal{S}|}. \nonumber
\end{align*}
We define $\mathcal{C}_\text{AC}(\theta)$ to be the \emph{autocorrection complexity}. 

The entropy production associated with autocorrection can only be minimized to zero if the engine is restricted to start in a single predictive state. Relative entropies are all non-negative and zero if and only if $\Pr(Y^{\theta'}_{0:\infty})=\Pr(Y^\theta_{0:\infty}|S_0=s_0)$. This means that for all predictive states $s_0$ with nonzero $p(s_0)$, to minimize entropy production it must be true that $\Pr(Y^{\theta'}_{0:\infty})=\Pr(Y^\theta_{0:\infty}|S_0=s_0)$. However, by the definition of causal states, different memory states must give rise to different future predictions. Thus, zero dissipation is impossible unless the initial distribution is a unique start state $p(s_0)= \delta_{s_0,s^*}$.

\subsection{Bayesian Edge-Weights}

We introduce another regularization technique by allowing the agent's memory to start in a single start state $p(s)=\delta_{s,s^*}$, but insisting on a complexity cost that is proportional to the work dissipated by initializing every edge-weight $C(\theta)= \langle W^\text{prepare}_\text{diss} \rangle$ ($\alpha=1$).  Applying Thm. \ref{Thm:edgeweights}, the resulting maximum-work edge-weights are given by:
\begin{align*}
\label{eq:EdgeWeights}
    \Theta^\text{max}_{\delta_{s,s^*},1,\epsilon,y_{0:L}}(y|s)& = \frac{1+N(y,s|s^*,\epsilon,y_{0:L})}{|\mathcal{Y}|+\sum_{y'}N(y',s|s^*,\epsilon,y_{0:L})}.
\end{align*}
This is simply Laplace's rule of succession when $Y_t$ is a binary variable \cite{jaynes2003probability}.  It follows from considering a uniform prior over the inputs for each causal state, then using Bayes' theorem to update the distribution using input counts $N(y,s|s^*,\epsilon,y_{0:L})$ for each causal state.  The generalization of this to larger alphabets $\mathcal{Y}$ also comes from Bayesian inference applied to a prior given by a $|\mathcal{Y}|$-dimensional uniform distribution (Dirichlet distribution) over the possible edge-weights \cite{jaynes2003probability}.

\subsection{Combined Regularization}

If we choose to both incur a cost of initializing the edge-weights ($\alpha=1$) and a cost of autocorrection ($p(s)=1/|\mathcal{S}|$), then we obtain a ``combined'' regularization strategy for which the maximum-work edge-weights are:
\begin{align}
    \Theta^\text{max}_{1/|\mathcal{Y}|,1,\epsilon,y_{0:L}}(y|s)& = \frac{\sum_{s_0} (1+N(y,s|s_0,\epsilon,y_{0:L}))}{\sum_{y',s_0} (1+N(y',s|s_0,\epsilon,y_{0:L}))}.
\end{align}
as shown in App. \ref{app:Maximum Work Edge Weights}.

The following numerically compares the advantages of each regularization strategy:
\begin{enumerate}
     \setlength{\topsep}{-5pt}
      \setlength{\itemsep}{-5pt}
      \setlength{\parsep}{-5pt}
\item MLE ($\alpha=0$, $p(s)=\delta_{s,s^*}$): Unregularized, where we simply maximize the work production, which yields the same inference as Maximum Likelihood Estimation.
\item BAYES ($\alpha=1$, $p(s)=\delta_{s,s^*}$): We incur the cost of initializing edge-weights, resulting in Bayesian updates of the edge-weights.
\item AC ($\alpha=0$, $p(s)=1/|\mathcal{S}|$): We incur the cost of autocorrection, meaning that divergent predictions from different causal states must be strongly justified by data.
\item CMBD ($\alpha=1$, $p(s)=1/|\mathcal{S}|$): We combine the costs of autocorrection and initializing edge-weights.
\end{enumerate}

\section{Testing Generalization Strategies}

\begin{figure*}[tbp]
\centering
\includegraphics[width=1.7\columnwidth]{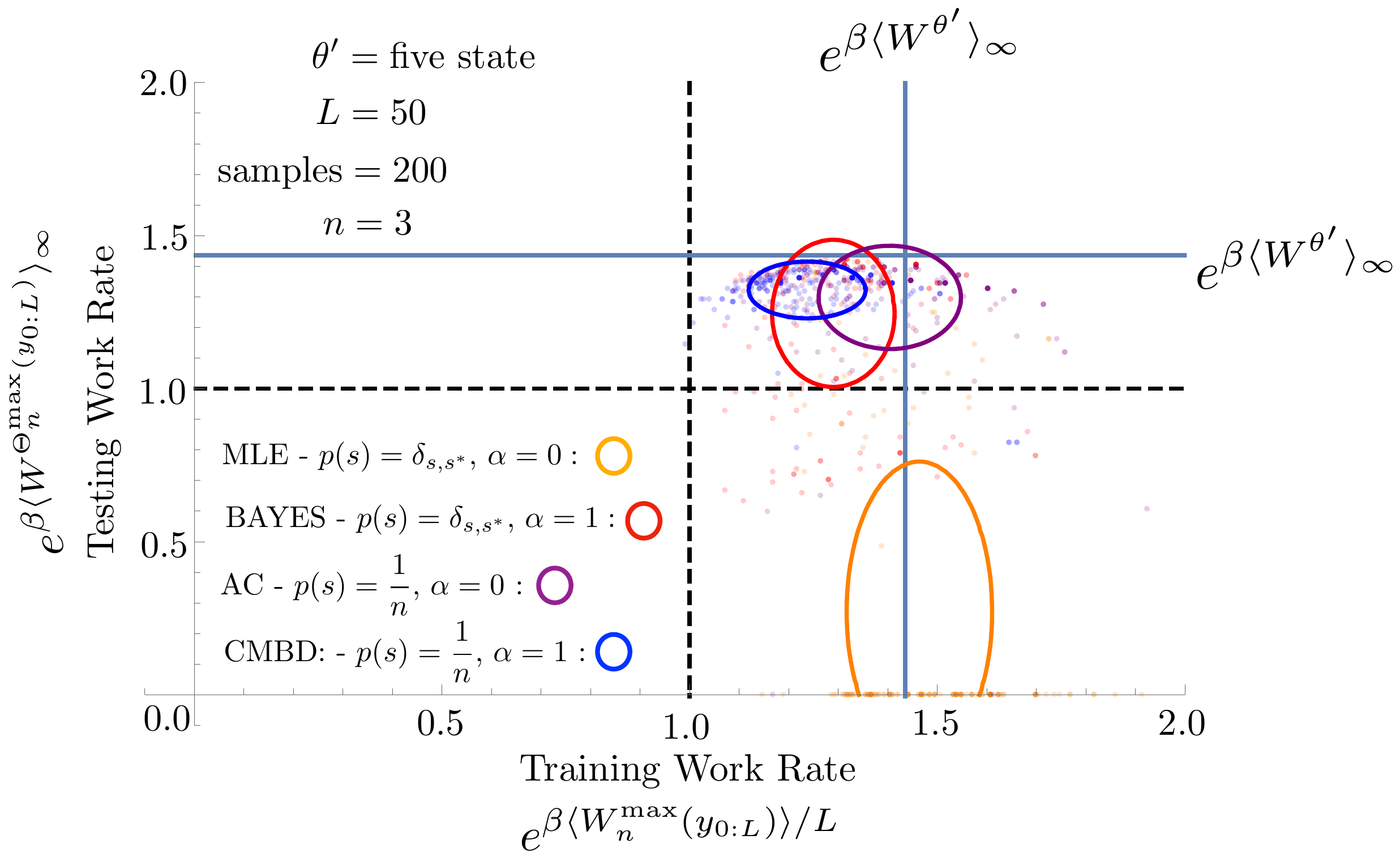}
\caption{Ensemble performance of thermodynamic regularization in training compared to testing: For the input process $\theta'$ we randomly sample $200$ length $L=50$ words and train 3-state information engines on each to find the exponential training work rate $e^{\beta \langle W^\text{max}_{n=3}(y_{0:L})\rangle/L}$ and the exponential testing work rate $e^{\beta \langle W^{\Theta^\text{max}_{n=3}(y_{0:L})}\rangle_\infty}$ for each regularization strategy: MLE (orange), BAYES (red), AC (purple), and CMBD (blue).  The ovals are centered around the average work rates of these $200$ samples, and their dimensions are given by the variance of the work rates.  The dashed black lines represent work rates of zero along each dimension, and the blue lines represent the theoretical limit on the asymptotic work rate, given by $e^{\beta \langle W^{\theta'} \rangle}$.  The strict MLE strategy performs best in training, but worst in testing.  The CMBD strategy performs worst in training and best in testing.}
\label{fig:SingleScatter}
\end{figure*}

As we attempt to determine the effectiveness of thermodynamic learning with an additional work penalty from complexity, it is worth noting that our original measure of testing performance for our model (the asymptotic work rate $\langle W^\theta \rangle_\infty$) is unchanged.  As long as the machines are ergodic, the initial distribution $p(s)$ doesn't affect the steady-state distribution $\pi_{s,s'}$ in Thm. \ref{thm:WorkRate}.  In addition, the complexity costs $C(\theta)$ are transient, and don't affect the asymptotic dynamics.  Thus, we can exactly calculate the asymptotic performance of regularized maximum-work models using Thm. \ref{thm:WorkRate}, just as before.  In the following, we evaluate the performance of the four different generalization strategies by calculating the maximum-work model
\begin{align}
    \Theta_n^\text{max}(y_{0:L})\equiv \underset{\theta \in \Theta_n}{\text{argmax}} \langle  W^\theta_G(y_{0:L}) \rangle,
\end{align}  
then evaluating the testing work rate $\langle W^{\Theta_n^\text{max}(y_{0:L})} \rangle_\infty$.

\begin{figure*}[tbp]
\centering
\includegraphics[width=2\columnwidth]{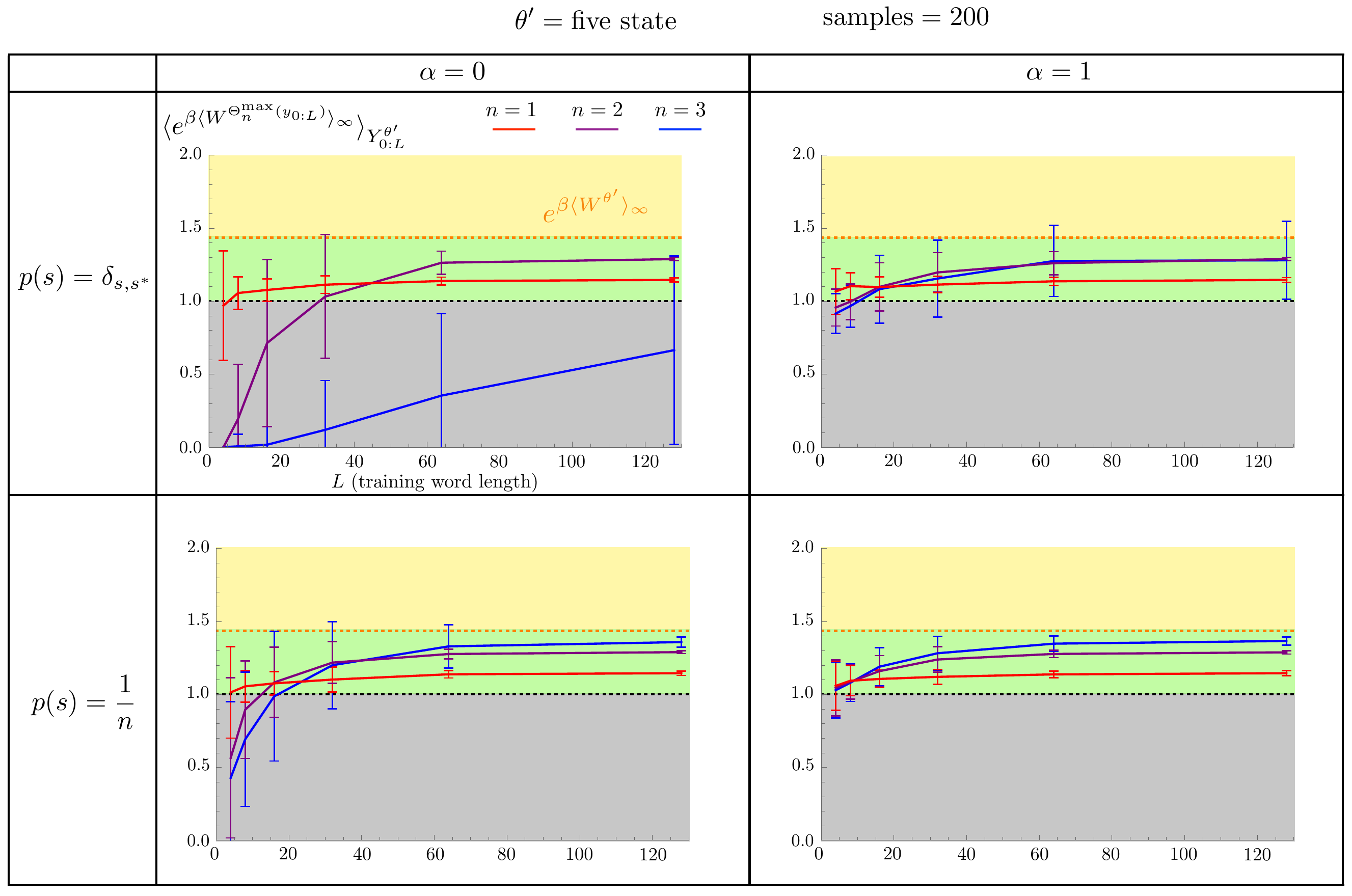}
\caption{CMBD regularization mitigates overfitting: The average and variance of the exponential asymptotic testing work rate that results from the four different learning strategies.  We generate $200$ words from the Five-State process for each length $L \in \{ 4,8,16,32,64,128\}$, then train on each using MLE, BAYES, AC, and CMBD.  The unregularized thermodynamic machine learning does an extremely poor job of discovering the underlying pattern in the word for large memories.  The BAYES method does well for small memories, but has poor performance for large memories, because of the large variance.  Using AC instead helps, but leads to divergent dissipation for small training sets.  For MLE, BAYES, and AC, there is an advantage to using a small amount of memory, for small data.  The CMBD strategy, by contrast, doesn't seem to have this suffer from using more memory, suggesting that it effectively mitigates overfitting for this case.}
\label{fig:FiveStateComparison}
\end{figure*}

\subsection{Ensemble Performance of Regularized Thermodynamic Learning}

The performance of thermodynamic learning from individual words as shown in Fig. \ref{fig:TrainingAndValidation} gives some insight into how patterns are discovered.  However, we see enough variety in random word realizations in App. \ref{app: Individual Words} that we must look at ensemble averages to determine whether Thermodynamic Regularization effectively generalizes. 

Take the process generated by the Five-State model shown in Fig. \ref{fig:ThreeMachines} as our true source.  This process requires five or more memory states for perfect prediction, while we have at most three memory states available.  Examining learning from this process elucidates the case where there is a force towards more engine complexity, but practical constraints of overfitting limit complexity.  This illustrates how thermodynamic learning performs when the accessible model class is misspecified by not including sufficient memory states to fully capture the process.  

Figure \ref{fig:SingleScatter} plots the asymptotic testing work rate against the training work rate for engines with memory size $n=3$ from $200$ different words of length $50$ generated from the Five-State model. Plotted against the theoretical limit on the work rate $\beta \langle W^{\theta'} \rangle_\infty$, we see how close the learned engines are to optimal.  Note that the work rates are exponentiated, so that the graph can accommodate infinitely divergent outcomes. We identify the behaviour of the ensemble of words by plotting an oval whose center is our numerical estimate of the average work rates for the learning process:
\begin{align*}
& (\langle x \rangle ,\langle y \rangle )
\\ & =\left(\langle e^{\beta \langle W^\text{max}_{n=3}(y_{0:L})\rangle/L} \rangle_{Y^{\theta'}_{0:L}}, \langle e^{\beta \langle W^{\Theta^\text{max}_{n=3}(y_{0:L})}\rangle_\infty} \rangle_{Y^{\theta'}_{0:L}}\right), \nonumber
\end{align*}
and whose radial dimensions are given by the variance of the exponential work rates:
\begin{widetext}
\begin{align*}
 (\text{var}( x )  ,\text{var}( y ) )=
  \left(\text{var}\left( e^{\beta \langle W^\text{max}_{n=3}(y_{0:L})/L\rangle} \right)_{Y^{\theta'}_{0:L}}, \text{var}\left( e^{\beta \langle W^{\Theta^\text{max}_{n=3}(y_{0:L})}\rangle_\infty}  \right)_{Y^{\theta'}_{0:L}}\right). 
\end{align*}
\end{widetext}
This uses the notation:
\begin{align*}
\langle f(x) \rangle_X &  \equiv \sum_{x} \Pr(X=x) f(x)
\\  \text{var}(f(x))_X & \equiv \sum_x \Pr(X=x) (f(x)-\langle f(x) \rangle)^2.
\end{align*}

To interpret Fig. \ref{fig:SingleScatter}, note that it plots the average and variance of the exponential work $e^W$ to accommodate cases of extreme overfitting, with infinite dissipation and $-\infty$ work.  If any elements of our training sample overfit in this way, then the resulting estimated average $\langle W \rangle$ and variance $\text{var}(W)$ of the work would diverge. For the case where $\alpha=0$, this is always a possibility, since a given training word may not realize a transition that is allowed by the input process.

We estimate a rough average from the plot via Jensen's inequality:
\begin{align*}
e^{\langle W \rangle} \leq \langle e^{W} \rangle.
\end{align*}
Equality is only satisfied when the work always realizes the average. This means that as the variance of the distribution increases, so should the difference below the average exponential $ \langle e^{W} \rangle$.  Thus, when reading the figure, both higher average and lower variance correspond to better engine performance.

We see that the training work rate is often above the theoretical limit set but $\langle W^{\theta'} \rangle_\infty$, but the testing work rate is always below it, as should be the case.  We also see that a higher training work rate seems to correspond to worse testing performance.  Beyond this, we focus on the testing work rate in the following analysis.

For the Five-State process, we see that the standard MLE technique without regularization badly overfits.  The average exponential asymptotic work rate after learning is well below unity. This reflects the fact that many realizations dissipate infinite work. The BAYES strategy, by contrast, improves on this work rate.  However, the variance is lower, and the average is higher for the AC and CMBD techniques, implying better test data performance for these two regularization strategies.  The CMBD strategy does best of all.

Figure \ref{fig:FiveStateComparison} shows the performance of different regularization strategies in greater detail by considering varying lengths.  For each length $L \in \{ 4,8,16,32,64,128\}$, we randomly generate $200$ words from the Five-State model and train on that model using MLE, BAYES, AC, and CMBD regularization. We also compare different memory sizes for the information engine.  We see that all strategies are relatively effective for engines with small memories.  However, for engines with three memory states, unregularized MLE results in divergent dissipation even for long training words.  We see better performance in BAYES and AC training, but for short words, we still see an advantage to training on models with less memory.  However, the CMBD technique seems to derive a consistent advantage from additional memory. This suggests that the CMBD regularization technique could be used to reliably discover complex patterns from small amounts of data.

Appendix \ref{app: Learning the Even Process and Noisy Even Process} goes into a similar analysis of the ``Even Process'' and ``Noisy Even Process''.  These are interesting since the true source is not misspecified by the candidate model class.  In both cases, since two memory states are sufficient to predict the process the $n=3$ curve lies at or below the $n=2$ for all samples, extracting less work on average.  Thus, we still see a cost to having an engine that is more complex than necessary for this two-state input.  However, CMBD regularization does the best in mitigating overfitting for larger engines.

\section{Summary}
Thermodynamics mandates both a drive towards complexity and an impetus for simplicity. In this article, we see this through the lens of thermodynamic machine learning (TML), which discovers patterns in data by maximizing work production.

Our framework implies a fundamental energy dissipation cost when an engine's internal model deviates from that of the true underlying distribution. To minimize this mismatch, an information engine must be at least as complex as the input - illustrating the principle of Requisite Complexity~\cite{Boyd17a, Boyd16d}.  Our first contribution is an exact expression for the asymptotic work rate of the information engine (see Thm. \ref{thm:WorkRate}), allowing us to directly evaluate the effectiveness of an engine in utilizing a particular information fuel.

Next, we greatly simplified the search for the maximum-work engine by analytically deriving the engine parameters in Thm. \ref{Thm:edgeweights}. This allowed us to implement TML over a wide class of engines, and demonstrate that solely maximizing work harvested from training data can lead to overfitting, with potentially dire energetic consequences (as seen in Fig. \ref{fig:TrainingAndValidation}). Moreover, this risk increases as engines scale in complexity. While the work production improves uniformly with engine memory during training, the resulting information engine anticipates phantom causal relations that are not present during testing. In extreme cases, such agents become immensely dissipative on unexpected inputs, and work dissipation can diverge towards infinity — the key signature of thermodynamic overfitting.

Our final contribution involves introducing two thermodynamically means of regularization.  The first adds a cost to initializing the information engine's edge-weights that is proportional to the dissipated work arising from relaxing to equilibrium with the prediction of each memory state. The second includes the cost of auto-correction - the energetic cost of discovering the correct predictive state when the engine starts out of sync with its information fuel. Combining the two techniques, we introduced analytical methods to construct regularized information engines that perform significantly better in harnessing free energy outside the training phase (see in Figs. \ref{fig:SingleScatter} and \ref{fig:FiveStateComparison}). In fact, our resulting engines perform similarly during testing and training phases, indicating that they have effectively mitigated overfitting. Meanwhile, the correspondence between work cost and maximum-likelihood estimation (MLE) together with the consistency of MLE~\cite{jaynes2003probability}, imply our agents converge to the correct model in the long-time limit. We thus illustrate how thermodynamic considerations can lead to effective means of pattern discovery without the risk of overfitting.

\section{Outlook}

Our results provide natural links between thermodynamics and machine learning and, in this way, establish a number of interesting connections worth further investigation. On a technical side, many of the results may be recast in the language of Fisher Information. Consider an ensemble of training words selected from the true distribution $\Pr(Y^{\theta'}_{0:L}=y_{0:L})$.  The variance in the estimator $\Theta_n^\text{max}(y_{0:L})$ can be bounded by using the Cramér-Rao bound. Asymptotically, this quantity scales linearly with word-length $L$, with the constant multiplier being the Fisher information rate \cite{radaelli2023fisher, Riech23a}. Maximum-likelihood estimators are asymptotically efficient, and so achieve this rate of learning with sufficient data \cite{lehmann2006theory}. Since the regularized TML methods presented here match MLE in the asymptotic limit, the variance of estimated parameters will follow the optimal $1/L$ scaling determined by the Fisher information rate.

Meanwhile, calculating the maximum-work engine edge-weights shown in Thm. \ref{Thm:edgeweights} strongly echoes reservoir computing in its computational simplicity \cite{hsu2023strange}. Reservoir computing leverages the inherent information processing of a complex system—the reservoir—driving it with time-series data. The only training that happens through this process occurs via an output layer, which homes in on the subdynamics of the reservoir that carry the relevant temporal correlations from the input sequence. This process is computationally simple, corresponding to a matrix inverse, with the weights of the output layers paralleling the edge-weights $\theta(y|s)$ of the predictive machine in TML. The question of what makes an effective reservoir remains open \cite{carroll2020reservoir}, with only heuristic design guides (e.g., good reservoirs are often thought to be on the ``edge of chaos'' \cite{schurmann2004edge}). It has been shown that the most effective reservoirs are deterministic (i.e., unifilar) \cite{marzen2018infinitely}. An engine's memory, viewed as a reservoir, satisfies this condition. Our framework may thus help provide thermodynamic guiding principles for finding effective reservoir computers.

More broadly, state-of-the-art time-series prediction and manipulation generally involve the use of recurrent neural networks (RNNs)~\cite{liao2018reviving} and transformers~\cite{vaswani2017attention}. In RNNs, the process of training is generally computationally intensive \cite{pascanu2013difficulty}, and recent works have suggested improved performance when such training makes use of causal discovery techniques in $\epsilon$-machines~\cite{zhang2019learning}. Meanwhile, transformer functionality is rooted in ``next token prediction'': finding a function that maps past inputs of some context length to a probability for the next ``token'' (input) \cite{kim2021code,o2021context}. Recent results showed that Transformers are indeed universal predictors, recovering a mapping from past inputs to hidden states and edge-weights \cite{basu2023transformers}. The $\epsilon$-map along with the edge-weights does just this, with the advantage that it does not require infinitely large memory to exactly model infinite Markov-order processes. Such processes would require infinite context-length for a Transformer to exactly predict them. By contrast, the memory states of a prediction engine can capture information contained in inputs arbitrarily far in the past with relatively small memory for many non-Markovian processes~\cite{marzen2016predictive}. As such, our methodologies may well provide new tools to tackle the unsustainable energetic cost of current AI models.

Finally, information is ultimately quantum-mechanical—leading to recurrent quantum models and quantum reservoir computers~\cite{gu2012quantum,gupta2021embodied,takaki2021learning}. Our thermodynamic toolkit thus provides a physical means to compare quantum and classical models operationally. Indeed, relations between work dissipation and imperfect modelling extend to the quantum regime~\cite{riechers2021}. Meanwhile, there is mounting evidence that such quantum machines can exhibit certain target behaviours in various contexts — stochastic modelling, string generation, adaptive strategies — while using less memory than any classical counterparts~\cite{binder2018practical,elliott2022quantum,vieira2022temporal,yano2021efficient,loomis2019strong}. In stochastic modelling, such memory advantages induce energetic advantages~\cite{loomis2020thermal}, while model memory can bound generalisation error in classification tasks~\cite{banchi2021generalization}. It would thus be exciting to determine how thermodynamic overfitting applies to quantum models, and thus determine whether the pressure for energetically efficient learning naturally motivates quantum-enhanced artificial intelligence.

We see an encouraging, evolving picture of how thermodynamic resources govern the emergence of predictive agents.  On the one hand, maximizing the energy extracted from information corresponds to discovering the maximum-likelihood predictive model of that information.  However, reckless energy extraction leads to overly-precise probability estimates, resulting in elevated energetic cost downstream.  Fortunately, physical constraints, such as the cost of instantiating the information engine and autocorrecting to the correct predictive state, prevent such glaring pitfalls.  It appears that nature conspires to bring about predictive machines through thermodynamic resource optimization.

\section*{Acknowledgments}
This work was supported by the Irish Research Council under grant number IRCLA/2022/3922, and by the Foundational Questions Institute and Fetzer Franklin Fund, a donor advised fund of the Silicon Valley Community Foundation, grant number FQXi-RFP-IPW-1910. ABB and JPC thank the Telluride Science Research Center for hospitality during visits and the participants of the Information Engines Workshops there. ABB acknowledges support from the Templeton World Charity Foundation Power of Information fellowships TWCF0337 and TWCF0560. This material is also based upon work supported by, or in part by, U.S. Army Research Laboratory, U.S. Army Research Office grant W911NF-21-1-0048, the National Research Foundation, Singapore, and Agency for Science, Technology and Research (A*STAR) under its QEP2.0 programme (NRF2021-QEP2-02-P06), and the Singapore Ministry of Education Tier 1 Grants RG146/20 and RG77/22, and the John Templeton Foundation grant no. 62423.

\clearpage
\onecolumngrid
\appendix

\section{Predictive Models}
\label{app:Review of Thermodynamic Machine Learning}

The predictive models encoded within an efficient information engine are Hidden Markov Models (HMMs). Specifically, they are edge-emitting such that they generate sequences through symbol-labeled transition matrices over hidden states $\mathcal{S}$:
\begin{align*}
    T^{(y)}_{s \rightarrow s'} \equiv \Pr(Y_i=y,S_{i+1}=s'|S_i=s).
\end{align*}
Predictive models are a subclass of HMMs such that the hidden states contain no information about the future beyond that which can be determined from the past: \cite{Boyd17a}
\begin{align*}
I[S_t;\overrightarrow{Y}_t|\overleftarrow{Y}_t]=0.
\end{align*}
The hidden states are called predictive states, because they contain all information in the past relevant for predicting the future.  If we additionally require that a process's model memory is minimal, then we obtain the $\epsilon$-machine: the minimal predictive generator of that process \cite{Crut12a}.  The predictive states of this minimal machine are often referred to as \emph{causal states}, because they describe the past's causal influence on the future.

The predictive complexity is described by the memory resources associated with the stationary causal state distribution $\Pr(S_t)$.  The $\epsilon$-machine is the predictive model whose hidden states minimize the $\alpha$-R\'enyi entropy $H_\alpha[S_t]$ for all $\alpha$ values \cite{loomis2019strong}.  $\alpha=1$ characterizes the statistical complexity \cite{Crut01a}:
\begin{align*}
    C_\mu & \equiv H_1[S_t]
    \\ & =-\sum_{s} \Pr(S_t=s) \ln \Pr(S_t=s), \nonumber
\end{align*} 
which is the channel capacity (measured in Nats) necessary to communicate from past to future.  By comparison, $\alpha=0$ yields is the  topological predictive complexity:
\begin{align*}
    H_0[S_t]= \ln |\mathcal{S}|,
\end{align*}
which is the log of the number of hidden states necessary to predictively model the process.  $H_0[S_t]$ is the measure that is most directly relevant in designing information engines, because it determines the number of memory states that an engine must have in order to efficiently harvest information.

The $\epsilon$-machine's minimal causal states are determined by a \emph{causal equivalence relation}~\cite{Crut88a}. As a result, they have the convenient property of \emph{unifilarity} \cite{Crut08b}, which means that the next causal state $S_{i+1}$ is uniquely specified by the current one $S_i$ and its output $Y_i$ by an $\epsilon$-map:
\begin{align*}
    S_{i+1}= \epsilon(Y_i,S_i).
\end{align*}
This is also known as the \emph{topology} of the $\epsilon$-machine \cite{strelioff2014bayesian}.  In addition, an $\epsilon$-machine has a unique start causal state~$s^*$.  In the case of a bi-infinite process $Y_{-\infty:\infty}$ which is stationary, this reflects the belief state of having seen nothing so far, but in the non-stationary semi-infinite $Y_{0:\infty}$ case, $s^*$ is simply the sufficient statistic of the (non-existent) past about the future~\cite{Boyd21a}.  

Unifilarity and a unique start state $s_0$ guarantee that an output sequence $y_{0:i}$ will lead to a unique causal state:
\begin{align*}
    s_i= \epsilon(y_{0:i},s_0),
\end{align*}
where we've re-used the update map notation, defining:
\begin{align*}
    \epsilon(y_{0:i},s_0) \equiv \epsilon ( y_{i-1},\cdots \epsilon(y_1,\epsilon(y_0,s_0)) \cdots ).
\end{align*}
Because of the isomorphism between the information engine and its $\epsilon$-machine model, the engine's memory state $x_i$ will be the same function of its input $x_i=\epsilon(y_{0:i},s^*)$.  Thus, we can see an efficient engine's memory dynamics as an input-conditioned deterministic dynamical system, paralleling efficient reservoir computers~\cite{marzen2018infinitely}.

Unifilarity means that we can characterize any $\epsilon$-machine in terms of its topology, start state, and edge-weights, as shown in the example of the Even Process in Fig. \ref{fig:ModelParameters}.  For a particular topology, the edge-weights $\theta(y|s)$ are the probability of outputting $y$ from causal state $s$:
\begin{align*}
    \theta(y|s) & \equiv T^{(y)}_{s \rightarrow \epsilon(y,s)}
    \\ & = \Pr(Y^{\theta}_i=y|S_i=s).
\end{align*}
This meaningfully simplifies the task of Thermodynamic Machine Learning.

Finally, $\epsilon$-machines can produce any process and yield the same likelihood as any predictive model, meaning that it is sufficient to limit our model class $\Theta$ to these models.  Such machines are described by a causal update map on the predictive states and edge-weights, meaning these are the explicit parameters we must explore through training.

\section{Asymptotic Work Rate: Derivation}
\label{app:Derivation of Asymptotic Work Rate}

The work produced by an efficient engine with model $\theta$ at the $L$th time step can be expressed:
\begin{align*}
    \beta \langle W^\theta \rangle_{L} &  \equiv \beta \langle W^\theta \rangle_{0:L+1}-\beta \langle W^\theta \rangle_{0:L}
    \\ & = \ln |\mathcal{Y}|+\sum_{y_{0:L+1}}\Pr(Y^{\theta'}_{0:L+1}=y_{0:L+1}) \ln \Pr(Y^{\theta}_{0:L+1}=y_{0:L+1}) -\sum_{y_{0:L}}\Pr(Y^{\theta'}_{0:L}=y_{0:L}) \ln \Pr(Y^\theta_{0:L}=y_{0:L}) 
    \\ & = \ln |\mathcal{Y}|+\sum_{y_{0:L+1}}\Pr(Y^{\theta'}_{0:L+1}=y_{0:L+1}) \ln \Pr(Y^{\theta}_{0:L+1}=y_{0:L+1}) -\sum_{y_{0:L+1}}\Pr(Y^{\theta'}_{0:L+1}=y_{0:L+1}) \ln \Pr(Y^\theta_{0:L}=y_{0:L}) 
    \\ & =\ln |\mathcal{Y}|+\sum_{y_{0:L+1}}\Pr(Y^{\theta'}_{0:L+1}=y_{0:L+1}) \ln \frac{\Pr(Y^{\theta}_{0:L+1}=y_{0:L+1})}{\Pr(Y^{\theta}_{0:L}=y_{0:L})}
    \\ & =\ln |\mathcal{Y}|+\sum_{y_{0:L+1}}\Pr(Y^{\theta'}_{0:L+1}=y_{0:L+1}) \ln \Pr(Y^{\theta}_L=y_L|Y^{\theta}_{0:L}=y_{0:L}).
\end{align*}
Using the fact that, for the model $\theta$, the edge-weights are related to output probabilities:
\begin{align*}
    \theta (y_L| \epsilon(y_{0:L}) )= \Pr(Y^{\theta}_L=y_L|Y^{\theta}_{0:L}=y_{0:L}) ,
\end{align*}
we can express the asymptotic work production:
\begin{align*}
    \beta \langle W^\theta \rangle_{L} =\ln |\mathcal{Y}|+\sum_{y_{0:L+1}}\Pr(Y^{\theta'}_{0:L+1}=y_{0:L+1}) \ln  \theta (y_L| \epsilon(y_{0:L}) ) .
\end{align*}
Another way of expressing this is to say the probability of the agent's causal state $s$ at time $L$ given the past inputs is $\Pr(S_L=s|Y^{\theta'}_{0:L}=y_{0:L})= \delta_{s, \epsilon(y_{0:L})}$, so that the work production can be expressed:
\begin{align*}
    \beta \langle W^\theta \rangle_{L} &  =\ln |\mathcal{Y}|+\sum_{y_{0:L+1}}\Pr(Y^{\theta'}_{0:L+1}=y_{0:L+1}) \ln  \theta (y_L| \epsilon(y_{0:L}) )
    \\ &  =\ln |\mathcal{Y}|+\sum_{y_{0:L+1},s}\Pr(Y^{\theta'}_{0:L+1}=y_{0:L+1}) \delta_{s, \epsilon(y_{0:L})} \ln \theta (y_L| \epsilon(y_{0:L}) )
    \\ &  =\ln |\mathcal{Y}|+\sum_{y_{0:L+1},s}\Pr(S_L=s,Y^{\theta'}_{0:L+1}=y_{0:L+1}) \ln \theta (y_L| s) .
\end{align*}
The causal states of the input's $\epsilon$-machine
are given by the random variables $S'_L$, so we express the asymptotic work:
\begin{align*}
    \beta \langle W^\theta \rangle_{L} &   =\ln |\mathcal{Y}|+\sum_{y_{0:L+1},s,s'}\Pr(S_L=s,S'_L=s',Y^{\theta'}_L=y_L,Y^{\theta'}_{0:L}=y_{0:L}) \ln \theta (y_L| s)
    \\ &   =\ln |\mathcal{Y}|+\sum_{y_{L},s,s'}\Pr(S_L=s,S'_L=s',Y^{\theta'}_L=y_L) \ln \theta (y_L| s)
    \\ &   =\ln |\mathcal{Y}|+\sum_{y,s,s'}\Pr(S_L=s,S'_L=s')\theta' (y_L| s')\ln \theta (y_L| s).
\end{align*}
If we take the asymptotic limit, we obtain the steady-state distribution over the causal states of the input $\epsilon$-machine driving the estimated $\epsilon$-machine:
\begin{align*}
    \pi_{s,s'} \equiv \lim_{L \rightarrow \infty}\Pr(S_L=s,S'_L=s'),
\end{align*}
which can be found by solving the steady-state equation:
\begin{align*}
\pi_{s_1,s_1'} = \sum_{s_0,s_0',y} \delta_{s_1,\epsilon(s_0,y)}  \delta_{s'_1,\epsilon'(s'_0,y)}  \theta' (y| s'_0) \pi_{s_0,s'_0}.
\end{align*}
Thus, we obtain the expression for the asymptotic work rate:
\begin{align*}
\beta \langle W^\theta \rangle_\infty =\ln |\mathcal{Y}| + \sum_{s,s',y} \pi_{s,s'} \theta'(y|s')  \ln \theta(y|s).
\end{align*}

\section{Work Production From Distributed Start State}
\label{app:Work Production From Distributed Start State}

Reference \cite{Boyd21a} primarily considers the work production from a particular start-state.  However, it also includes an equation for the work production when the hidden state starts in a distribution $\Pr(X_0=x_0)$, given in Eq. (G1):
\begin{align*}
    \beta \langle W^\theta (y_{0:L}) \rangle & = \sum_{x_0} \Pr(X_0=x_0) \ln \prod_{i=0}^{L-1}\theta(y_i|\epsilon(y_{0:i},x_0)) + L \ln |\mathcal{Y}|+\sum_{x_0}\Pr(X_0=x_0) \ln \frac{Pr(X^\theta_0=x_0)}{Pr(X^\theta_L=\epsilon(y_{0:L},x_0))},
\end{align*}
where $\Pr(X^\theta_i)$ is the estimated distribution of the agent memory $\mathcal{X}$ at time $i$, and $\Pr(X_i)$ is the actual distribution at the same time.  This is the average work that is produced from the initial memory distribution $\Pr(X_0)$ if one applies an efficient information engine based on the model $\theta$.  

However, we should note that operating on the input string will transform the memory distribution from $\Pr(X_0=x_0)$ to $\Pr(X_L=x_L)=\sum_{y_{0:L},x_0}\Pr(X_0=x_0,Y_{0:L}=y_{0:L}) \delta_{x_L,\epsilon(y_{0:L},x_0)}$, which the agent will estimate as transforming from $\Pr(X^\theta_0=x_0)$ to $\Pr(X^\theta_L=x_L)=\sum_{y_{0:L},x_0}\Pr(X^\theta_0=x_0,Y^\theta_{0:L}=y_{0:L}) \delta_{x_L,\epsilon(y_{0:L},x_0)}$.  To reset the protocol, we will reset to our estimated initial memory distribution $\Pr(X^\theta_0)$, which will involve an efficient transformation for which (according to Thm. 1 of Ref. \cite{Boyd21a}) the work production for input and output will be:
\begin{align*}
    \langle W(x_L \rightarrow x') \rangle = \ln \frac{\Pr(X^\theta_L=x_L)}{\Pr(X^\theta_0=x')}.
\end{align*}
Note that at the beginning of the reset, the distribution over $\mathcal{X}$ is given by:
\begin{align*}
    \Pr(X^\theta_L=x_L|Y^\theta_{0:L}=y_{0:L}) \equiv \sum_{x_0}\Pr(X_0=x_0) \delta_{x_L,\epsilon(y_{0:L},x_0)}.
\end{align*}
After the reset, the distribution is $\Pr(X^\theta_0=x')$, and because we quasistatically evolve to the post-reset distribution, the final distribution is independent of the distribution before reset, meaning that the average work production is:
\begin{align*}
& \sum_{x_L,x'}\Pr(X^\theta_L=x_L|Y^\theta_{0:L}=y_{0:L})\Pr(X^\theta_0=x') \langle W(x_L \rightarrow x') \rangle
\\ & =\sum_{x_L,x'}\Pr(X^\theta_L=x_L|Y^\theta_{0:L}=y_{0:L})\Pr(X^\theta_0=x') \ln \Pr(X_L^\theta=x_L) -\sum_{x_L,x'}\Pr(X^\theta_L=x_L|Y^\theta_{0:L}=y_{0:L})\Pr(X^\theta_0=x') \ln \Pr(X^\theta_0=x')
\\ & =\sum_{x_L}\Pr(X^\theta_L=x_L|Y^\theta_{0:L}=y_{0:L}) \ln \Pr(X_L^\theta=x_L) -\sum_{x'}\Pr(X^\theta_0=x') \ln \Pr(X^\theta_0=x')
\\ & =\sum_{x_L}\sum_{x_0}\Pr(X_0=x_0) \delta_{x_L,\epsilon(y_{0:L},x_0)} \ln \Pr(X_L^\theta=x_L) -\sum_{x'}\Pr(X^\theta_0=x') \ln \Pr(X^\theta_0=x')
\\ & =\sum_{x_0}\Pr(X_0=x_0)  \ln \Pr(X_L^\theta=\epsilon(y_{0:L},x_0)) -\sum_{x_0}\Pr(X^\theta_0=x_0) \ln \Pr(X^\theta_0=x_0).
\end{align*}
When we add this reset cost to the work benefit of harvesting energy, we find the total work production:
\begin{align*}
     \beta \langle W^\theta (y_{0:L}) \rangle  = \sum_{x_0} \Pr(X_0=x_0) \Pr(Y^\theta_{0:L}=y_{0:L}|S_0=x_0)+ L \ln |\mathcal{Y}| +\sum_{x_0}(\Pr(X_0=x_0)-\Pr(X^\theta_0=x_0)) \ln \Pr(X^\theta_0=x_0),
\end{align*}
where the probability of input $y_{0:L}$ given the initial state causal state is:
\begin{align*}
    \Pr(Y^\theta_{0:L}=y_{0:L}|S_0=x_0)=\prod_{i=0}^{L-1}\theta(y_i|\epsilon(y_{0:i},x_0)).
\end{align*}

Let us posit that the agent correctly estimates the initial distribution to be $p(s_0)$, such that:
\begin{align*}
p(x_0) = \Pr(X_0=x_0)= \Pr(X^\theta_0=x_0),
\end{align*}
either because it prepares the initial memory distribution, or has reliably measured it in the past.  Finally, since $\mathcal{X}=\mathcal{S}$, we can express the work production on this distributed state:
\begin{align*}
     \beta \langle W^\theta (y_{0:L}) \rangle  = \sum_{x_0}p(s_0) \Pr(Y^\theta_{0:L}=y_{0:L}|S_0=s_0)+ L \ln |\mathcal{Y}|,
\end{align*}
which is purely a function of underlying $\epsilon$-machine model $\theta$.

\section{Maximum Work Edge-Weights}
\label{app:Maximum Work Edge Weights}

For a given engine topology $\epsilon$, a particular initial distribution $p(s)$, and an additional energy cost of initializing the machine $C(\theta)$, the work production from a particular input string $y_{0:L}$ is:
\begin{align*}
    \beta \langle W^\theta_G (y_{0:L}) \rangle & = L \ln |\mathcal{Y}| -C(\theta) + \sum_{s_0}p(s_0) \ln \Pr(Y^\theta_{0:L}=y_{0:L}|S_0=s_0) 
    \\ & = L \ln |\mathcal{Y}|-C(\theta) +\sum_{s} p(s_0) \ln \prod_{i=0}^{L-1} \theta(y_i|\epsilon(y_{0:i},s_0)) .
\end{align*}
Here, the subscript $G$ indicates that this work production is modified in an attempt to make maximum-work training generalize.  We find the edge-weights for this topology by counting the input-memory state combinations of the engine when driven by $y_{0:L}$.  $N(y,s|s_0,\epsilon,y_{0:L})$ is the number of times predictive memory state $s$ receives input $y$ given that an engine with topology $\epsilon$ started in $s_0$ and received input word $y_{0:L}$.  This allows us to rewrite the work production:
\begin{align*}
    \beta \langle W^\theta_G (y_{0:L}) \rangle & = L \ln |\mathcal{Y}|-C(\theta) + \sum_{s_0}p(s_0)  \ln \prod_{i=0}^{L-1} \theta(y_i|\epsilon(y_{0:i},s_0))
    \\ & =L \ln |\mathcal{Y}| -C(\theta) +\sum_{s_0} p(s_0) \sum_{s',y'}  N(y,s|s^*,\epsilon,y_{0:L})\ln \theta(y'|s').
\end{align*}

We can find the resulting maximum-work edge-weights by taking the set of constraints:
\begin{align*}
g_{s'}(\theta)\equiv \sum_{y'}\theta(y'|s')=1,
\end{align*}
and solving:
\begin{align*}
    \partial_{\theta(y|s)}\beta \langle W^\theta_G (y_{0:L}) \rangle & = \sum_{s'}\lambda_{s'} \partial_{\theta(y|s)}g_{s'}(\theta)
    \\ -\partial_{\theta(y|s)}C(\theta)+\sum_{s_0,s',y'} p(s_0) \frac{N(y,s|s_0,\epsilon,y_{0:L})}{\theta(y'|s')}\delta_{sy,s'y'} & = \sum_{s'}\lambda_{s'} \sum_{y'}\delta_{sy,s'y'}
     \\ -\partial_{\theta(y|s)}C(\theta)+ \frac{\sum_{s_0} p(s_0)N(y,s|s_0,\epsilon,y_{0:L})}{\theta(y|s)} & = \lambda_{s} 
     \\ \frac{\sum_{s_0} p(s_0)N(y,s|s_0,\epsilon,y_{0:L})-\theta(y|s)\partial_{\theta(y|s)}C(\theta)}{\lambda_s} & = \theta(y|s)
     \\ \frac{\sum_{s_0} p(s_0)(N(y,s|s_0,\epsilon,y_{0:L})-\theta(y|s)\partial_{\theta(y|s)}C(\theta))}{\lambda_s} & = \theta(y|s).
\end{align*}
The constraint of normalized edge-weights $\sum_y \theta(y|s)$ allows us to solve for $\lambda_s$ as the normalization constant:
\begin{align*}
    \sum_{y',s_0} p(s_0)(N(y,s|s_0,\epsilon,y_{0:L})-\theta(y'|s)\partial_{\theta(y'|s)}C(\theta))=\lambda_s.
\end{align*}
The resulting maximum-work edge-weights $\Theta^\text{max}_{p,\epsilon,y_{0:L}}(y|s)$  are the solution to the recursive relation:
\begin{align*}
   \theta(y|s) = \frac{\sum_{s_0} p(s_0)(N(y,s|s_0,\epsilon,y_{0:L})-\theta(y|s)\partial_{\theta(y|s)}C(\theta))}{\sum_{y',s_0} p(s_0)(N(y',s|s_0,\epsilon,y_{0:L})-\theta(y'|s)\partial_{\theta(y'|s)}C(\theta))}.
\end{align*}

Here, we consider the cost:
\begin{align*}
    C(\theta) = - \alpha \sum_{y',s'} \ln \theta(y'|s')+\text{const.}, 
\end{align*}
such that:
\begin{align*}
-\partial_{\theta(y|s)}C(\theta) = \frac{\alpha}{\theta(y|s)},
\end{align*}
and we can remove the dependence on the right-hand side of the equation. Our resulting solution for edge-weights is:
\begin{align*}
\label{eq:appEdgeWeights}
    \Theta^\text{max}_{p,\alpha,\epsilon,y_{0:L}}(y|s)& = \frac{\sum_{s_0} p(s_0)(\alpha+N(y,s|s_0,\epsilon,y_{0:L}))}{\sum_{y',s_0} p(s_0)(\alpha+N(y',s|s_0,\epsilon,y_{0:L}))}.
\end{align*}

Note that if we are implementing the transformation quasistatically, the distribution $\theta(y|s)$ must be the equilibrium distribution over $\mathcal{Y}$ when conditioned on the engine memory state $s$.  If we have the initial energy landscape $E(y,s)$, then the equilibrium distribution is:
\begin{align*}
    \pi(y,s) = e^{\beta(F^\text{eq}-E(y,s))},
\end{align*}
meaning:
\begin{align*}
    \theta(y|s) & =\frac{ \pi^\theta(y,s)}{\pi^\theta(s)}
    \\ & = \frac{e^{-\beta E^\theta(y,s)}}{\sum_y e^{-\beta E^\theta(y,s)}} ,
\end{align*}
where $\pi^\theta(s) \equiv \sum_y \pi^\theta(y,s)$ is the marginal equilibrium distribution over the agent memory.

Note that we can also write the metastable free energy \cite{Parr15a, riechers2020balancing, wimsatt2022trajectory},  of the memory state $s$:
\begin{align*}
    \beta F^\theta(s) = - \ln \sum_{y}e^{-\beta E^\theta(y,s)}.
\end{align*}
We can express the edge-weights in terms of the free energy of memory state $s$ and the initial energy of:
\begin{align*}
 \theta(y|s) = e^{\beta(F^\theta(s)-E^\theta(s,y))}.
\end{align*}
Thus, the cost can be expressed as the thermodynamic quantity:
\begin{align*}
    C(\theta)&  = - \alpha \sum_{y',s'} \ln \theta(y'|s')
    \\ & = \alpha \sum_{y',s'} \beta(E^\theta(s,y)-F^\theta(s)).
\end{align*}
If we incur a cost that is proportional to the total energetic excess beyond the free energy for each memory state, then we will find the solution for the edge-weights given by  $\Theta^\text{max}_{p,\alpha,\epsilon,y_{0:L}}(y|s)$.

This can be recovered by considering the cost of initializing every edge-weight and relaxing into equilibrium with every predictive state $s$. Mechanically, this relates to the fact that the energy of a state is related to its equilibrium distribution:
\begin{align*}
E(z) = F^\text{eq}-\ln \Pr(Z^\text{eq}=z).
\end{align*}
At the beginning of a quasistatic protocol, the energy landscape must be in equilibrium with the estimated distribution, meaning that:
\begin{align*}
E(y,s) =F^\text{eq}-\ln \theta(y|s)\pi(s),
\end{align*}
where $\pi(s)= \Pr(S^\text{eq}=s)$.  For each state, we must initialize it.  This corresponds to confining the state to $y$ and $s$, then letting it relax to equilibrium with output distribution $\theta(y|s)$.  The amount of work dissipated in this relaxation is the change in nonequilibrium addition to free energy, which is the relative entropy:
\begin{align*}
\langle W_\text{diss}(y,s)\rangle & = D_{KL}( \delta_{y,y'} \delta_{s,s'}|| \theta(y'|s')\pi(s))-D_{KL}( \theta(y'|s'')\delta_{s,s'}|| \theta(y'|s')\pi(s'))
\\ & = \sum_{s'y'} \delta_{y,y'} \delta_{s,s'} \ln \frac{\delta_{y,y'} \delta_{s,s'}}{ \theta(y'|s')\pi(s')}- \sum_{s'y'} \theta(y'|s')\delta_{s,s'} \ln \frac{\theta(y'|s')\delta_{s,s'}}{ \theta(y'|s')\pi(s')}
\\ & = - \ln \theta(y|s) \pi(s) - \sum_{y'} \theta(y'|s) \ln \frac{1}{ \pi(s)}
\\ & = - \ln \theta(y|s) \pi(s) +\ln  \pi(s)
\\ & = - \ln \theta(y|s).
\end{align*}
If we incur this dissipation for every combination of predictive state and input, then the total dissipated work is:
\begin{align*}
\langle W^\text{prepare}_\text{diss} \rangle =  -\sum_{s,y} \ln \theta(y|s).
\end{align*}
Thus, our cost is proportional to our original cost function $C(\theta)$, meaning that we can set:
\begin{align*}
    C(\theta) = \alpha  \langle W^\text{prepare}_\text{diss} \rangle,
\end{align*}
to solve for the the same maximum-work edge-weights $\Theta^\text{max}_{p,\alpha,\epsilon,y_{0:L}}(y|s)$ described in Eq. (\ref{eq:appEdgeWeights}).

We investigate four cases:
\begin{enumerate}
     \setlength{\topsep}{-5pt}
      \setlength{\itemsep}{-5pt}
      \setlength{\parsep}{-5pt}
    \item Without any attempt at generalization, we choose $\alpha=0$ and a peaked initial distribution $p(s_0)=\delta_{s_0,s^*}$, then we obtain the standard likelihood expression for work:
    \begin{align}
    \beta \langle W^\theta_G (y_{0:L}) \rangle & = L \ln |\mathcal{Y}| +  \ln \Pr(Y_{0:L}=y_{0:L}|S_0=s^*) 
    \\ & = L \ln |\mathcal{Y}|+\ln \prod_{i=0}^{L-1} \theta(y_i|\epsilon(y_{0:i},s^*)) \nonumber .
\end{align}
which produces the familiar MLE estimate for the edge-weights:
\begin{align*}
    \Theta^\text{max}_{\delta_{s,s^*},0,\epsilon,y_{0:L}}(y|s)& = \frac{N(y,s|s^*,\epsilon,y_{0:L})}{\sum_{y'} N(y',s|s^*,\epsilon,y_{0:L})}.
\end{align*}
\item If we attempt to generalize through autocorrection, we choose a uniform initial state $p(s)=1/|\mathcal{S}|$ and zero contribution from the complexity cost $\alpha=0$, yielding the work production:
    \begin{align*}
    & \beta \langle W^\theta_G (y_{0:L}) \rangle  = L \ln |\mathcal{Y}| + \frac{1}{|\mathcal{S}|} \sum_{s_0}  \ln \Pr(Y_{0:L}=y_{0:L}|S_0=s_0).
\end{align*}
The modified edge-weight estimator includes the contributions from every start state:
\begin{align}
    \Theta^\text{max}_{1/|\mathcal{S}|,0,\epsilon,y_{0:L}}(y|s)& = \frac{\sum_{s_0} N(y,s|s_0,\epsilon,y_{0:L})}{\sum_{y',s_0} N(y',s|s_0,\epsilon,y_{0:L})}.
\end{align}

\item We might also try to generalize by only adding a cost from the dissipation of initializing the system at $\alpha=1$, while allowing a unique start state $p(s)=\delta_{s,s^*}$, such that the work production is:
\begin{align*}
    \beta \langle W^\theta_G (y_{0:L}) \rangle  & =   \ln \Pr(Y_{0:L}=y_{0:L}|S_0=s^*) + L \ln |\mathcal{Y}| - |\mathcal{Y}| \sum_{s'}\beta \langle W^\theta_\text{diss} (s') \rangle. \nonumber
\end{align*}
The resulting maximum-work estimator is:
\begin{align*}
    \Theta^\text{max}_{\delta_{s,s^*},1,\epsilon,y_{0:L}}(y|s)& = \frac{1+N(y,s|s^*,\epsilon,y_{0:L})}{|\mathcal{Y}|+\sum_{y'}N(y',s|s^*,\epsilon,y_{0:L})}.
\end{align*}
This is precisely Laplace's rule of succession applied to each causal state, which is the result of Bayesian updating from a uniform distribution over output probabilities.

\item Last, we combine the complexity cost $\alpha=1$ with the cost of autocorrection $p(s)=1/|\mathcal{S}|$ such that the work production is:
\begin{align*}
    \beta \langle W^\theta_G (y_{0:L}) \rangle & = \frac{1}{|\mathcal{Y}|}\sum_{s} \ln \Pr(Y_{0:L}=y_{0:L}|S_0=s_0)+ L \ln |\mathcal{Y}|-|\mathcal{Y}| \sum_{s'}\beta \langle W^\theta_\text{diss} (s') \rangle.
\end{align*}
The maximum-work edge-weights combine the benefits of both generalization strategies:
\begin{align*}
    \Theta^\text{max}_{1/|\mathcal{Y}|,1,\epsilon,y_{0:L}}(y|s)& = \frac{\sum_{s_0} (1+N(y,s|s_0,\epsilon,y_{0:L}))}{\sum_{y',s_0} (1+N(y',s|s_0,\epsilon,y_{0:L}))}.
\end{align*}

\section{Entropy Production as Divergence}
\label{app: Entropy Production as Divergence}

We monitor engine inefficiency via entropy production (dissipated work) \cite{Parr15a}. This is the net work invested minus the change in nonequilibrium free energy:
\begin{align*}
    \langle \Sigma^\theta \rangle_{0:L}/k_B = \beta (-\langle W^\theta \rangle_{0:L}-\Delta F^\text{NEQ}_{0:L}).
\end{align*}
Because an information reservoir has equal energy for all configurations, the contributions to the nonequilibrium free energy are only informational:
\begin{align*}
    \beta \Delta F^\text{NEQ}_{0:L} & =-\Delta H_{0:L}
    \\ & = H[Y^{\theta'}_{0:L}]-L \ln |\mathcal{Y}|
\end{align*}
where $H[Z]\equiv - \sum_{z} \Pr(Z=z) \ln \Pr(Z=z) $ is the Shannon entropy of of random variable $Z$ measured in Nats.  The average work, by contrast, is:
\begin{align*}
\langle W^\theta \rangle_{0:L} = L \ln |\mathcal{Y}| + \sum_{y_{0:L}} \Pr(Y^{\theta'}_{0:L}=y_{0:L}) \ln \Pr(Y^{\theta}_{0:L}=y_{0:L}).
\end{align*}
Adding these terms together, we get the relative entropy between the true input process $Y^{\theta'}_{0:L}$ and the estimated process $Y^{\theta}_{0:L}$:
\begin{align*}
    \langle \Sigma^\theta \rangle_{0:L}/k_B & = \sum_{y_{0:L}} \Pr(Y^{\theta'}_{0:L}=y_{0:L}) \ln \frac{\Pr(Y^{\theta'}_{0:L}=y_{0:L})}{\Pr(Y^{\theta}_{0:L}=y_{0:L})}
    \\ & \equiv D_{KL}(Y^{\theta'}_{0:L}||Y^{\theta}_{0:L}), \nonumber
\end{align*}
which is the additional dissipation that is incurred by misestimating the input distribution \cite{Kolc17a, Riec20a}.

If a learning process refines and improves the estimator $\theta$, its divergence from the actual process $Y^{\theta'}_{0:L}$ should diminish, reflecting the idea that learning reduces entropy production \cite{milburn2023quantum, gold2019self}.  The consistency of MLE guarantees that, as long as the true process can be described by one of the available models, the learning process will discover the true distribution \cite{Cove06a, jaynes2003probability}.  Thus, given a sufficiently large class of $\epsilon$-machines to select from, thermodynamic machine learning will discover the hidden process and minimize the average entropy production to zero for future inputs.

For an information engine harvesting energy from an information reservoir, the asymptotic rate of change in free energy is the difference between the entropy rates of the input process and output process \cite{Boyd15a}:
\begin{align*}
    \beta \Delta F^\text{NEQ}_\infty & \equiv \lim_{L \rightarrow \infty} \beta (\Delta F^\text{NEQ}_{0:L+1}-\Delta F^\text{NEQ}_{0:L})
    \\ & = h_\mu^{\theta'}- \ln |\mathcal{Y}|.
\end{align*}
Here, the entropy rate $h_\mu^{\theta'}$ of the inputs can be directly calculated from its $\epsilon$-machine \cite{Crut01a}:
\begin{align*}
    h^{\theta'}_\mu = -\sum_{s',y} \pi'_{s'} \theta'(y|s') \ln \theta'(y|s'),
\end{align*}
where $\pi_s' = \sum_{s} \pi_{s,s'}$ is the steady-state distribution of the true $\epsilon$-machine's causal states.  As a result, the asymptotic entropy production rate can be expressed as the average divergence between the edge-weights:
\begin{align*}
    \langle \Sigma^\theta \rangle_\infty/k_B & = \beta (-\langle W^\theta \rangle_\infty-\Delta F^\text{NEQ}_\infty)
    \\ & = \sum_{s,s',y} \pi_{s,s'} \theta'(y|s') \ln \frac{\theta'(y|s')}{\theta(y|s)} \nonumber
    \\ & = \sum_{s,s'} \pi_{s,s'}  D_{KL} (Y^{\theta'}_i|S'_i=s'||Y^{\theta}_i|S_i=s). \nonumber
\end{align*}
Here:
\begin{align*}
 D_{KL} (Y^{\theta'}_i|S'_t=s'||Y^{\theta}_i|S_i=s)  \equiv \sum_y  \Pr(Y^{\theta'}_i=y|S'_i=s') \ln \frac{ \Pr(Y^{\theta'}_i=y|S'_i=s')}{\Pr(Y^{\theta}_i=y|S_i=s)} \nonumber
\end{align*}
is the divergence between the prediction of the next input from state memory state $s$ in model $\theta$ and the prediction from memory state $s'$ in model $\theta'$.

\section{Training and Testing Individual Words}
\label{app: Individual Words}

As an illustrative example, consider two training words of length $100$ generated from the Five-State machine shown in Fig. \ref{fig:ThreeMachines}.  Figure \ref{fig:FitExamples} shows the resulting (exponentiated) maximum training work production rate $e^{\beta \langle W^\text{max}_n(y_{0:L}) \rangle/L}$ (solid lines) for memory sizes $n \in \{1,2,3\}$ when we train on length-$1$ to length-$100$ for two different training words.  The dashed lines show the (exponentiated) asymptotic work production rate $e^{\beta \langle W^{\Theta^\text{max}_n(y_{0:L})} \rangle_\infty}$, representing our performance in testing.  

In the upper-left corner of each diagram, we see the result of un-regularized work-maximization ($\alpha=0$ and $p(s)=\delta_{s,s^*}$), which yields improved work production in training with greater memory.  But, it dangerously overfits, dissipating divergent work for large model memories.  

By contrast, we see in the upper right corners that when $\alpha=1$ and $p(s)=\delta_{s,s^*}$, leading to Bayesian estimates of edge-weights, the model doesn't divergently overfit.  For some training words (word 1 of Fig. \ref{fig:FitExamples}), we see that the model learns how to extract almost all available work for the large-memory case ($n=3$).  However, other training words (word 2 of Fig. \ref{fig:FitExamples}) show that, even though the model doesn't divergently overfit, it dissipates considerable energy, producing negative work on average for the large-memory case.

The case of autocorrection in the bottom left, when  $\alpha=0$ and $p(s)=1/|\mathcal{S}|$, leads to error mitigation as well.  There is still divergent dissipation for very short training words shown in Fig. \ref{fig:FitExamples}, but the thermodynamic learning appears to discover useful patterns in the input data, with a growing advantage for larger memories.  This suggests that autocorrection mitigates overfitting for complex processes.

Finally, in the bottom right ($\alpha=1$ and $p(s)=1/|\mathcal{S}|$), where we have combined regularization strategies to both require autocorrection and a complexity cost in initializing the model, we see that overfitting appears to be largely mitigated.  Generalizing in this way appears to allow us to add memory without adding considerable risk of overfitting.  This is promising, as we would like to be able to perform thermodynamic learning with much more complex $\epsilon$-machine models to discover complex patterns in data without the risk of projecting structure onto the data that isn't there.

On the whole, it appears that the result of training varies considerably depending on the training word.  This is as expected for short word lengths, as the same word can come from many different processes. For this reason, we now explore in more detail as we vary across ensembles of training words.

\begin{figure*}[tbp]
\centering
\includegraphics[width=\columnwidth]{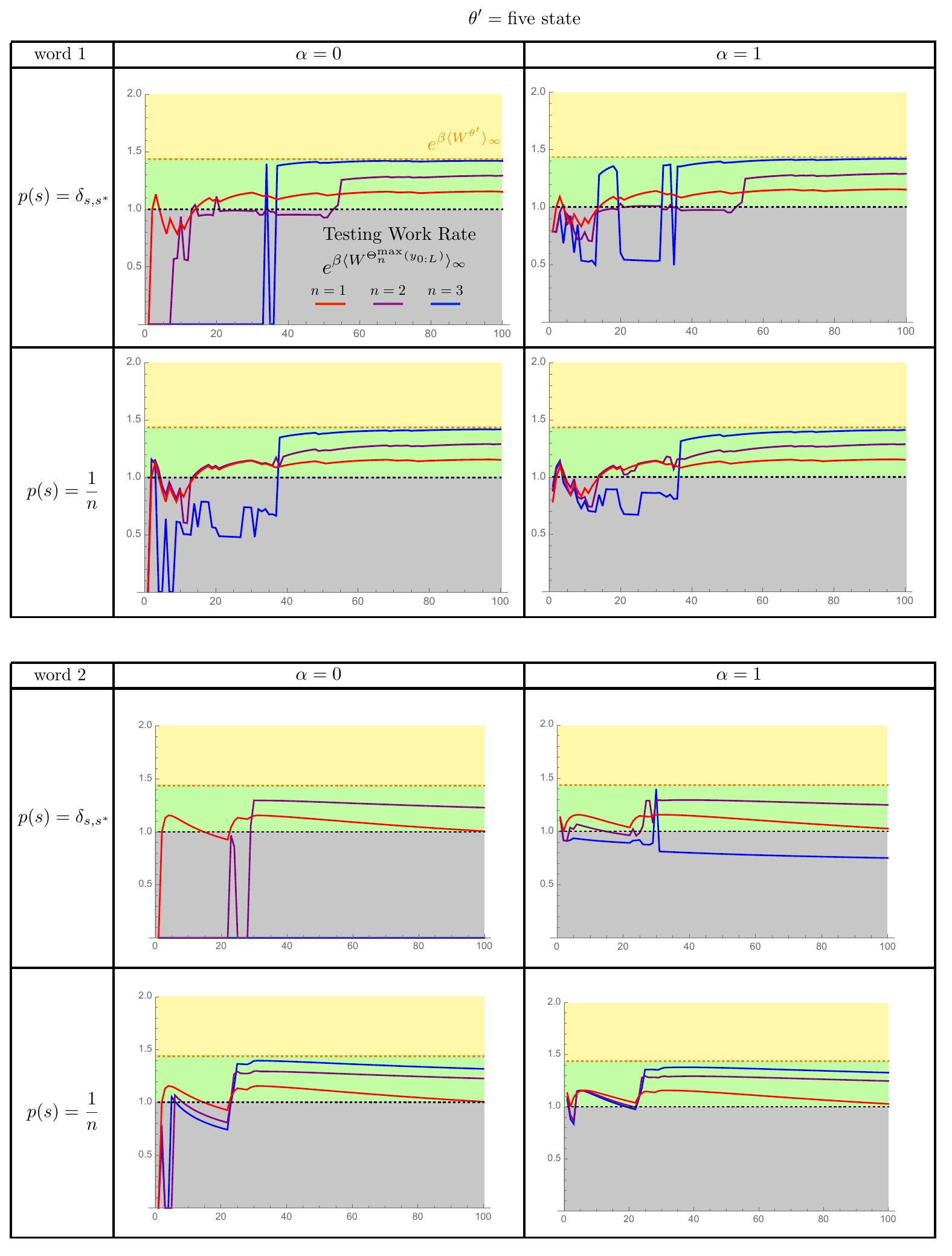}
\caption{The top and bottom correspond to two different training words (``word 1'' and ``word 2'') from the ``Five-State'' $\epsilon$-machine.  We plot testing work rate with four different regularization strategies: 1) Unregularized TML ($\alpha=0$ and $p(s)=\delta_{s,s^*}$)  2) Autocorrection ($\alpha=0$ and $p(s)=1/|\mathcal{S}|$) 3) Bayesian complexity cost ($\alpha=1$ and $p(s)=\delta_{s,s^*}$) 4) Combined ($\alpha=1$ and $p(s)=1/|\mathcal{S}|$).}
\label{fig:FitExamples}
\end{figure*}

\section{Learning the Even Process and Noisy Even Process}
\label{app: Learning the Even Process and Noisy Even Process}

\begin{figure*}[tbp]
\centering
\includegraphics[width=\columnwidth]{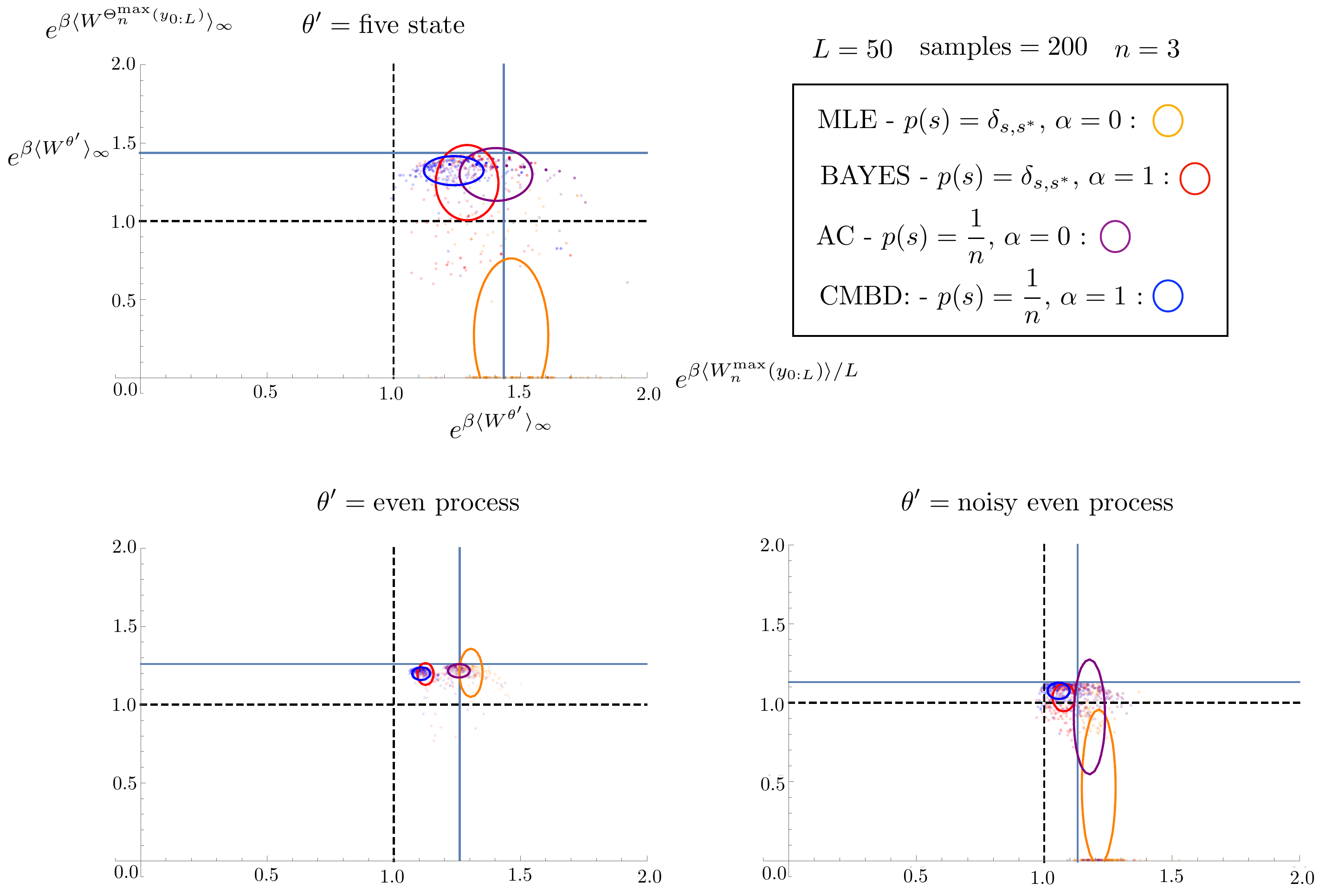}
\caption{Thermodynamic Machine Learning Performance for Different Processes: For each input process $\theta'$ (Five-State, Even, and Noisy Even), we randomly sample $200$ length $L=50$ words and train a 3-state $\epsilon$-machines on each to find the exponential training work rate $e^{\beta \langle W^\text{max}_{n=3}(y_{0:L})\rangle/L}$ and the exponential testing work rate $e^{\beta \langle W^{\Theta^\text{max}_{n=3}(y_{0:L})}\rangle_\infty}$ for each regularization strategy: MLE (orange), BAYES (red), AC (purple), and CMBD (blue).  The ovals are centered around the average work rates of these $200$ samples, and their dimensions are given by the variance of the work rates.  The dashed black lines represent work rates of zero along each dimension, and the blue lines represent the theoretical limit on the asymptotic work rate, given by $e^{\beta \langle W^{\theta'} \rangle}$.}
\label{fig:PerformanceScatter}
\end{figure*}

\begin{figure*}[tbp]
\centering
\includegraphics[width=\columnwidth]{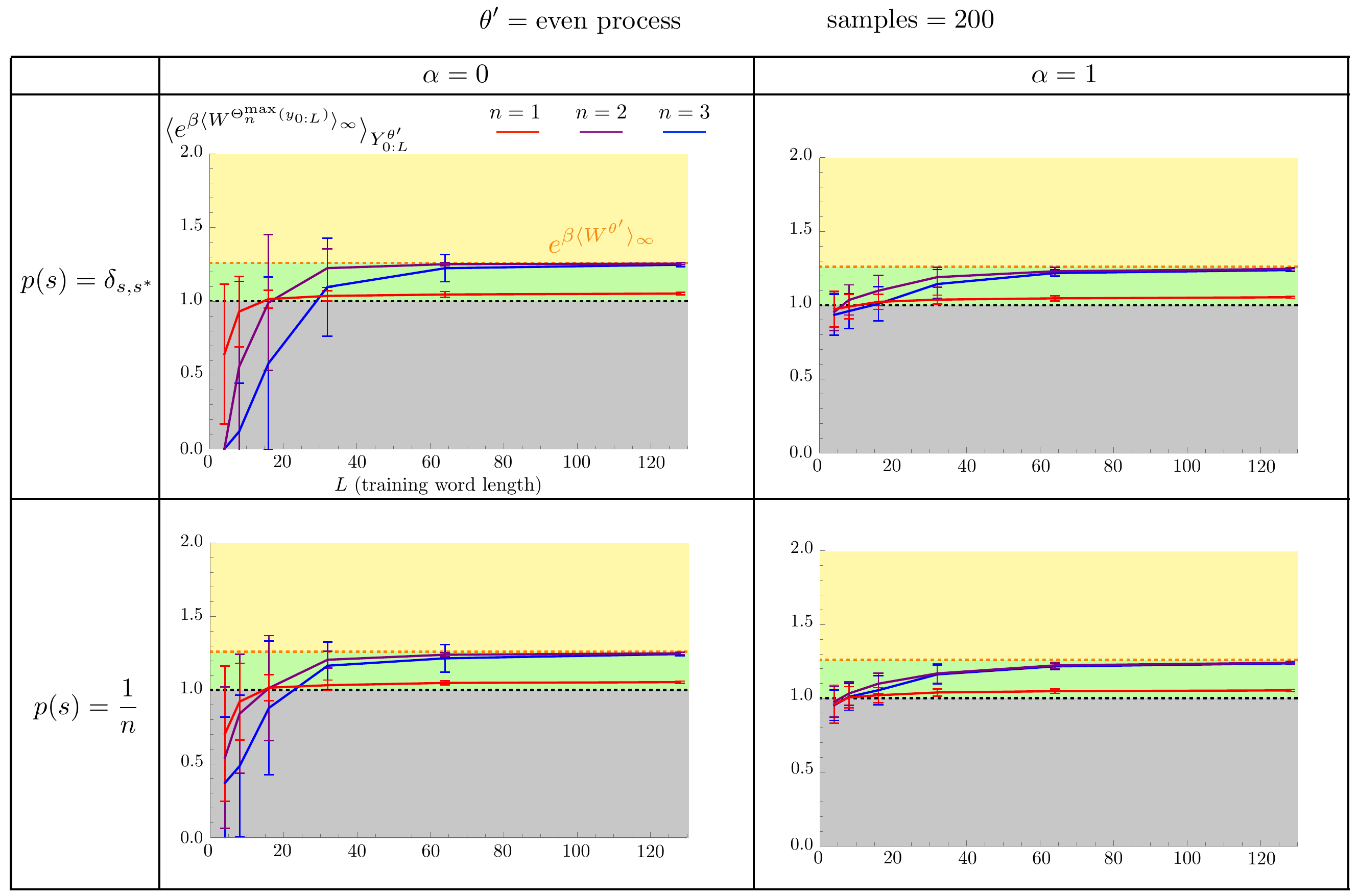}
\caption{The average and variance of the exponential asymptotic testing work rate that results from the four different learning strategies.  We generate $200$ words from the even process for each length $L \in \{ 4,8,16,32,64,128\}$, then train on each using MLE, BAYES, AC, and CMBD.  The un-regularized thermodynamic machine learning does a fairly good job of fitting, in comparison to the Five-State process.  The AC method is a slight improvement over MLE, but not considerable.  The BAYES and CMBD regularization techniques have less divergent work production for small training words, and seem to perform slightly better in general.  One notable distinction between the BAYES and CMBD methods is that there is a clear disadvantage to using larger memory than necessary ($n=3$ instead of $n=2$) for the BAYES technique, but the CMBD technique does not seem to have this difference.}
\label{fig:EvenComparison}
\end{figure*}

\begin{figure*}[tbp]
\centering
\includegraphics[width=\columnwidth]{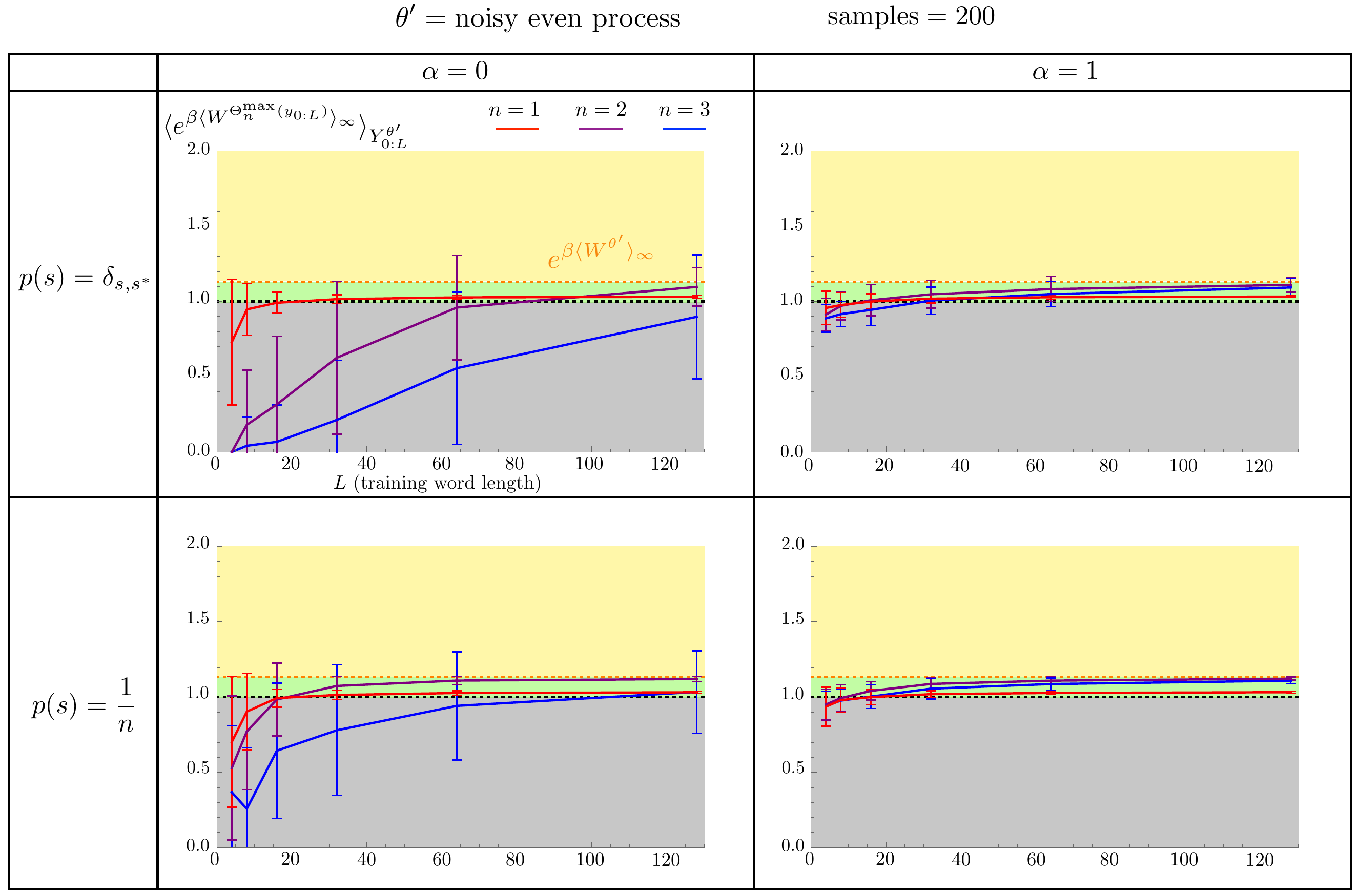}
\caption{The average and variance of the exponential asymptotic testing work rate that results from the four different learning strategies.  We generate $200$ words from the noisy even process for each length $L \in \{ 4,8,16,32,64,128\}$, then train on each using MLE, BAYES, AC, and CMBD. This process appears much harder to learn than the even process for the unregularized MLE and AC techniques.  The BAYES technique does better, even after 128 training words, the variance in the outcome is still relatively large, indicating that many samples are failing to effectively learn the process.  By contrast, the CMBD technique appears to approach the upper bound on work production, with variance that's relatively small for memory sizes $n=2$ and $n=3$.}
\label{fig:NoisyEvenComparison}
\end{figure*}

In contrast to the ``Five-State Process,''  the ``Even Process'' and ``Noisy Even Process'' require only two states for perfect prediction, so they reveal the case when we have more memory accessible ($n=3$) than is strictly necessary.  In this case, the true model is contained within our class of model candidates.  

The ``Noisy Even Process'' is an interesting case for the same reason, but it has full support in the possible words that it can produce.  A good learning algorithm should robustly learn patterns that may contain rare fluctuations and noise, and the ``Noisy Even Process'' behaves much like the ``Even process,'' but its rare fluctuations are extremely important to thermodynamic behavior \cite{crutchfield2016not}.

In Fig \ref{fig:PerformanceScatter}, we compare the training work rate to the testing work rate for $200$ samples for all three processes.  We see in all cases that the MLE strategy works best for training data and words for test data, while the CMBD method works worst for training and best for testing.  However, the relative effectiveness of the BAYES and AC method changes based on which true process is being sampled.  This suggests that autocorrection serves to address some features of processes, while Bayesian edge-weight updates serve to address others.

The even process appears to be easy to learn for all cases, as seen in Fig. \ref{fig:EvenComparison}.  The unregularized MLE strategy appears to do worst, because its variance is the highest while its average is comparable to the regularized learning techniques.  BAYES, AC, and CMBD all appear to have comparable results for the testing work rate, with BAYES perhaps performing the best.  In this case, unlike for the Five-State process, we get to see what happens when our model has enough memory to fully capture the pattern in the data.  Interestingly, additional memory then comes at a cost for all regularization techniques, which can be identified by noting that the blue curve ($n=3$) lies at or below the purple curve ($n=2$).  This may be because of the fact that allowing three memory states simply creates more opportunity for overfitting when the process is fully described by two predictive states.  However, the CMBD strategy seems to mostly mitigate the overfitting, because the blue curve is very close to the purple one.

Figure \ref{fig:NoisyEvenComparison} shows that the noisy even process appears to be much more difficult to learn than the even process.  We hypothesize that this is because we need to include words that are rare fluctuations to fully see the support of the process.  When a model doesn't anticipate those words, it results in overfitting and negative work production.  We see that the unregularized MLE and AC techniques encounter this problem frequently.  By contrast, the BAYES and CMBD strategies seem to discover the underlying pattern after training on length-$128$.  However, the BAYES technique still yields high variance, meaning that some outcomes are dissipating a lot of work.  By contrast, the CMBD algorithm has low variance and its average is nearly the theoretical limit on work harvesting, indicating that it is reliably discovering the pattern.

\end{enumerate}
\clearpage
\twocolumngrid
\bibliography{short}

\begin{thebibliography}{10}

\bibitem{dongare2012introduction}
AD~Dongare, RR~Kharde, Amit~D Kachare, et~al.
\newblock ``Introduction to artificial neural network''.
\newblock International Journal of Engineering and Innovative Technology
  (IJEIT) {\bf 2}, 189--194~(2012).
\newblock  url:~\url{https://api.semanticscholar.org/CorpusID:212457035}.

\bibitem{sohl2015deep}
Jascha Sohl-Dickstein, Eric Weiss, Niru Maheswaranathan, and Surya Ganguli.
\newblock ``Deep unsupervised learning using nonequilibrium thermodynamics''.
\newblock In Proceedings of the 32nd International Conference on Machine
  Learning.
\newblock Volume~37, pages 2256--2265.
\newblock PMLR~(2015).
\newblock  url:~\url{https://proceedings.mlr.press/v37/sohl-dickstein15.html}.

\bibitem{bahri2020statistical}
Yasaman Bahri, Jonathan Kadmon, Jeffrey Pennington, Sam~S Schoenholz, Jascha
  Sohl-Dickstein, and Surya Ganguli.
\newblock ``Statistical mechanics of deep learning''.
\newblock
  \href{https://dx.doi.org/10.1146/annurev-conmatphys-031119-050745}{Annual
  Review of Condensed Matter Physics {\bf 11}, 501--528}~(2020).

\bibitem{melanson2023thermodynamic}
Denis Melanson, Mohammad~Abu Khater, Maxwell Aifer, Kaelan Donatella,
  Max~Hunter Gordon, Thomas Ahle, Gavin Crooks, Antonio~J Martinez, Faris
  Sbahi, and Patrick~J Coles.
\newblock ``Thermodynamic computing system for ai applications''~(2023).
\newblock  \href{http://arxiv.org/abs/2312.04836}{arXiv:2312.04836}.

\bibitem{coles2023thermodynamic}
Patrick~J Coles, Collin Szczepanski, Denis Melanson, Kaelan Donatella,
  Antonio~J Martinez, and Faris Sbahi.
\newblock ``Thermodynamic ai and the fluctuation frontier''.
\newblock In 2023 IEEE International Conference on Rebooting Computing (ICRC).
\newblock \href{https://dx.doi.org/10.1109/ICRC60800.2023.10386858}{Pages
  1--10}.
\newblock IEEE~(2023).

\bibitem{Fris10a}
Karl Friston.
\newblock ``{The free-energy principle : a unified brain theory?}''.
\newblock \href{https://dx.doi.org/10.1038/nrn2787}{Nature Reviews Neuroscience
  {\bf 11}, 127--138}~(2010).

\bibitem{ali2022predictive}
Abdullahi Ali, Nasir Ahmad, Elgar de~Groot, Marcel Antonius~Johannes van
  Gerven, and Tim~Christian Kietzmann.
\newblock ``Predictive coding is a consequence of energy efficiency in
  recurrent neural networks''.
\newblock \href{https://dx.doi.org/10.1016/j.patter.2022.100639}{Patterns {\bf
  3}, 100661}~(2022).

\bibitem{gupta2021embodied}
Agrim Gupta, Silvio Savarese, Surya Ganguli, and Li~Fei-Fei.
\newblock ``Embodied intelligence via learning and evolution''.
\newblock
  \href{https://dx.doi.org/https://doi.org/10.1038/s41467-021-25874-z}{Nature
  communications {\bf 12}, 5721}~(2021).

\bibitem{bishop2006pattern}
Christopher Bishop.
\newblock ``{Pattern Recognition and Machine Learning}''.
\newblock Springer. ~(2006).
\newblock
  url:~\url{https://www.microsoft.com/en-us/research/publication/pattern-recognition-machine-learning/}.

\bibitem{Boyd21a}
Alexander~B Boyd, James~P Crutchfield, and Mile Gu.
\newblock ``Thermodynamic machine learning through maximum work production''.
\newblock \href{https://dx.doi.org/10.1088/1367-2630/ac4309}{New Journal of
  Physics {\bf 24}, 083040}~(2021).

\bibitem{Stil12a}
Susanne Still, David~A. Sivak, Anthony~J. Bell, and Gavin~E. Crooks.
\newblock ``{Thermodynamics of Prediction}''.
\newblock \href{https://dx.doi.org/10.1103/PhysRevLetters109.120604}{Physical
  Review Letters {\bf 109}, 120604}~(2012).

\bibitem{Crut88a}
J.~P. Crutchfield and K.~Young.
\newblock ``Inferring statistical complexity''.
\newblock \href{https://dx.doi.org/10.1103/PhysRevLetters63.105}{Physical
  Review Letters {\bf 63}, 105--108}~(1989).

\bibitem{Crut01a}
J.~P. Crutchfield and D.~P. Feldman.
\newblock ``Regularities unseen, randomness observed: Levels of entropy
  convergence''.
\newblock \href{https://dx.doi.org/10.1063/1.1530990}{CHAOS {\bf 13},
  25--54}~(2003).

\bibitem{Crut12a}
J.~P. Crutchfield.
\newblock ``Between order and chaos''.
\newblock \href{https://dx.doi.org/10.1038/NPHYS2190}{Nature Physics {\bf 8},
  17--24}~(2012).

\bibitem{gold2019self}
Jacob~M Gold and Jeremy~L England.
\newblock ``Self-organized novelty detection in driven spin glasses''~(2019).

\bibitem{ying2019overview}
Xue Ying.
\newblock ``An overview of overfitting and its solutions''.
\newblock \href{https://dx.doi.org/10.1088/1742-6596/1168/2/022022}{Journal of
  physics: Conference series {\bf 1168}, 022022}~(2019).

\bibitem{Boyd17a}
A.~B. Boyd, D.~Mandal, and J.~P. Crutchfield.
\newblock ``Thermodynamics of modularity: Structural costs beyond the landauer
  bound''.
\newblock \href{https://dx.doi.org/10.1103/PhysRevX.8.031036}{Physical Review X
  {\bf 8}, 031036}~(2018).

\bibitem{nakkiran2021deep}
Preetum Nakkiran, Gal Kaplun, Yamini Bansal, Tristan Yang, Boaz Barak, and Ilya
  Sutskever.
\newblock ``Deep double descent: Where bigger models and more data hurt''.
\newblock \href{https://dx.doi.org/10.1088/1742-5468/ac3a74}{Journal of
  Statistical Mechanics: Theory and Experiment {\bf 2021}, 124003}~(2021).

\bibitem{tibshirani1996regression}
Robert Tibshirani.
\newblock ``Regression shrinkage and selection via the lasso''.
\newblock \href{https://dx.doi.org/10.1111/j.2517-6161.1996.tb02080.x}{Journal
  of the Royal Statistical Society Series B: Statistical Methodology {\bf 58},
  267--288}~(1996).

\bibitem{zhang2010regularization}
Yiyun Zhang, Runze Li, and Chih-Ling Tsai.
\newblock ``Regularization parameter selections via generalized information
  criterion''.
\newblock \href{https://dx.doi.org/10.1198/jasa.2009.tm08013}{Journal of the
  American statistical Association {\bf 105}, 312--323}~(2010).

\bibitem{tanaka2019recent}
Gouhei Tanaka, Toshiyuki Yamane, Jean~Benoit H{\'e}roux, Ryosho Nakane, Naoki
  Kanazawa, Seiji Takeda, Hidetoshi Numata, Daiju Nakano, and Akira Hirose.
\newblock ``Recent advances in physical reservoir computing: A review''.
\newblock
  \href{https://dx.doi.org/https://doi.org/10.1016/j.neunet.2019.03.005}{Neural
  Networks {\bf 115}, 100--123}~(2019).

\bibitem{liao2018reviving}
Renjie Liao, Yuwen Xiong, Ethan Fetaya, Lisa Zhang, KiJung Yoon, Xaq Pitkow,
  Raquel Urtasun, and Richard Zemel.
\newblock ``Reviving and improving recurrent back-propagation''.
\newblock In International Conference on Machine Learning.
\newblock Volume~80, pages 3082--3091.
\newblock PMLR~(2018).
\newblock  url:~\url{https://proceedings.mlr.press/v80/liao18c.html}.

\bibitem{vaswani2017attention}
Ashish Vaswani, Noam Shazeer, Niki Parmar, Jakob Uszkoreit, Llion Jones,
  Aidan~N Gomez, {\L}ukasz Kaiser, and Illia Polosukhin.
\newblock ``Attention is all you need''.
\newblock Advances in neural information processing systems {\bf 30},
  5998--6008~(2017).
\newblock
  url:~\url{https://proceedings.neurips.cc/paper_files/paper/2017/file/3f5ee243547dee91fbd053c1c4a845aa-Paper.pdf}.

\bibitem{Maxw71}
J.~C. Maxwell.
\newblock ``Theory of heat''.
\newblock Longmans, Green and Co. London, United Kingdom~(1871).

\bibitem{Thom74a}
W.~Thomson.
\newblock ``Kinetic theory of the dissipation of energy''.
\newblock Nature {\bf 9}, 441--444~(1874).

\bibitem{Leff2002}
Harvey Leff and Andrew~F. Rex, editors.
\newblock ``{Maxwell's Demon 2: Entropy, Classical and Quantum Information,
  Computing}''.
\newblock \href{https://dx.doi.org/10.1201/9781420033991}{CRC Press}. ~(2002).

\bibitem{Szilard1929}
Leo Szilard.
\newblock ``{\"U}ber die entropieverminderung in einem thermodynamischen system
  bei eingriffen intelligenter wesen''.
\newblock \href{https://dx.doi.org/10.1007/BF01341281}{Zeitschrift f{\"u}r
  Physik {\bf 53}, 840--856}~(1929).

\bibitem{Land61a}
R.~Landauer.
\newblock ``Irreversibility and heat generation in the computing process''.
\newblock IBM J. Res. Develop. {\bf 5}, 183--191~(1961).

\bibitem{Parr15a}
J.~M.~R. Parrondo, J.~M. Horowitz, and T.~Sagawa.
\newblock ``Thermodynamics of information''.
\newblock \href{https://dx.doi.org/https://doi.org/10.1038/nphys3230}{Nature
  Physics {\bf 11}, 131--139}~(2015).

\bibitem{esposito2010entropy}
Massimiliano Esposito, Katja Lindenberg, and Christian Van~den Broeck.
\newblock ``Entropy production as correlation between system and reservoir''.
\newblock \href{https://dx.doi.org/10.1088/1367-2630/12/1/013013}{New Journal
  of Physics {\bf 12}, 013013}~(2010).

\bibitem{Reeb2014}
David Reeb and Michael~M Wolf.
\newblock ``{An improved Landauer principle with finite-size corrections}''.
\newblock \href{https://dx.doi.org/10.1088/1367-2630/16/10/103011}{New Journal
  of Physics {\bf 16}, 103011}~(2014).

\bibitem{Taranto2023}
Philip Taranto, Faraj Bakhshinezhad, Andreas Bluhm, Ralph Silva, Nicolai Friis,
  Maximilian~P.E. Lock, Giuseppe Vitagliano, Felix~C. Binder, Tiago Debarba,
  Emanuel Schwarzhans, Fabien Clivaz, and Marcus Huber.
\newblock ``{Landauer Versus Nernst: What is the True Cost of Cooling a Quantum
  System}''.
\newblock \href{https://dx.doi.org/10.1103/PRXQuantum.4.010332}{PRX Quantum
  {\bf 4}, 010332}~(2023).

\bibitem{Croo99a}
G.~E. Crooks.
\newblock ``Entropy production fluctuation theorem and the nonequilibrium work
  relation for free energy differences''.
\newblock
  \href{https://dx.doi.org/https://doi.org/10.1103/PhysRevE.60.2721}{Physical
  Review E {\bf 60}, 2721}~(1999).

\bibitem{esposito2011second}
Massimiliano Esposito and Christian Van~den Broeck.
\newblock ``Second law and landauer principle far from equilibrium''.
\newblock \href{https://dx.doi.org/10.1209/0295-5075/95/40004}{Europhysics
  Letters {\bf 95}, 40004}~(2011).

\bibitem{seifert2005entropy}
Udo Seifert.
\newblock ``Entropy production along a stochastic trajectory and an integral
  fluctuation theorem''.
\newblock
  \href{https://dx.doi.org/https://doi.org/10.1103/PhysRevLett.95.040602}{Physical
  Review Letters {\bf 95}, 040602}~(2005).

\bibitem{garner2017thermodynamics}
Andrew~JP Garner, Jayne Thompson, Vlatko Vedral, and Mile Gu.
\newblock ``Thermodynamics of complexity and pattern manipulation''.
\newblock
  \href{https://dx.doi.org/https://doi.org/10.1103/PhysRevE.95.042140}{Physical
  Review E {\bf 95}, 042140}~(2017).

\bibitem{touzo2020optimal}
L{\'e}o Touzo, Matteo Marsili, Neri Merhav, and {\'E}dgar Rold{\'a}n.
\newblock ``Optimal work extraction and the minimum description length
  principle''.
\newblock \href{https://dx.doi.org/10.1088/1742-5468/abacb3}{Journal of
  Statistical Mechanics: Theory and Experiment {\bf 2020}, 093403}~(2020).

\bibitem{Mand012a}
D.~Mandal and C.~Jarzynski.
\newblock ``Work and information processing in a solvable model of {Maxwell's}
  demon''.
\newblock
  \href{https://dx.doi.org/https://doi.org/10.1073/pnas.1204263109}{Proc. Natl.
  Acad. Sci. USA {\bf 109}, 11641--11645}~(2012).

\bibitem{Boyd15a}
A.~B. Boyd, D.~Mandal, and J.~P. Crutchfield.
\newblock ``Identifying functional thermodynamics in autonomous {Maxwellian}
  ratchets''.
\newblock \href{https://dx.doi.org/doi:10.1088/1367-2630/18/2/023049}{New
  Journal Physics {\bf 18}, 023049}~(2016).

\bibitem{strasberg2015thermodynamics}
Philipp Strasberg, Javier Cerrillo, Gernot Schaller, and Tobias Brandes.
\newblock ``Thermodynamics of stochastic turing machines''.
\newblock
  \href{https://dx.doi.org/https://doi.org/10.1103/PhysRevE.92.042104}{Physical
  Review E {\bf 92}, 042104}~(2015).

\bibitem{lin2017does}
Henry~W Lin, Max Tegmark, and David Rolnick.
\newblock ``Why does deep and cheap learning work so well?''.
\newblock
  \href{https://dx.doi.org/https://doi.org/10.1007/s10955-017-1836-5}{Journal
  of Statistical Physics {\bf 168}, 1223--1247}~(2017).

\bibitem{upper1997theory}
Daniel~Ray Upper.
\newblock ``Theory and algorithms for hidden markov models and generalized
  hidden markov models''.
\newblock University of California, Berkeley. ~(1997).

\bibitem{strelioff2014bayesian}
Christopher~C Strelioff and James~P Crutchfield.
\newblock ``Bayesian structural inference for hidden processes''.
\newblock \href{https://dx.doi.org/10.1103/PhysRevE.89.042119}{Physical Review
  E {\bf 89}, 042119}~(2014).

\bibitem{adhikary2020expressiveness}
Sandesh Adhikary, Siddarth Srinivasan, Geoff Gordon, and Byron Boots.
\newblock ``Expressiveness and learning of hidden quantum markov models''.
\newblock In International Conference on Artificial Intelligence and
  Statistics.
\newblock Pages 4151--4161.
\newblock PMLR~(2020).
\newblock  url:~\url{https://proceedings.mlr.press/v108/adhikary20a.html}.

\bibitem{basu2023transformers}
Sourya Basu, Moulik Choraria, and Lav~R Varshney.
\newblock ``Transformers are universal predictors''~(2023).
\newblock  \href{http://arxiv.org/abs/2307.07843}{arXiv:2307.07843}.

\bibitem{claeskens2008model}
Gerda Claeskens, Nils~Lid Hjort, et~al.
\newblock ``Model selection and model averaging''.
\newblock Volume 330.
\newblock Cambridge University Press Cambridge. ~(2008).

\bibitem{vieira2022temporal}
Lucas~B Vieira and Costantino Budroni.
\newblock ``Temporal correlations in the simplest measurement sequences''.
\newblock
  \href{https://dx.doi.org/https://doi.org/10.22331/q-2022-01-18-623}{Quantum
  {\bf 6}, 623}~(2022).

\bibitem{jurgens2021shannon}
Alexandra~M Jurgens and James~P Crutchfield.
\newblock ``Shannon entropy rate of hidden markov processes''.
\newblock
  \href{https://dx.doi.org/https://doi.org/10.1007/s10955-021-02769-3}{Journal
  of Statistical Physics {\bf 183}, 32}~(2021).

\bibitem{Kolc17a}
A.~Kolchinsky and D.~H. Wolpert.
\newblock ``Dependence of dissipation on the initial distribution over
  states''.
\newblock \href{https://dx.doi.org/10.1088/1742-5468/aa7ee1}{J. Stat. Mech.:
  Th. Expt. {\bf 2017}, 083202}~(2017).

\bibitem{Riec20a}
P.~M. Riechers and M.~Gu.
\newblock ``Initial-state dependence of thermodynamic dissipation for any
  quantum process''.
\newblock \href{https://dx.doi.org/10.1103/PhysRevE.103.042145}{Physical Review
  EPage 042145}~(2020).
\newblock  \href{http://arxiv.org/abs/2002.11425}{arXiv:2002.11425}.

\bibitem{milburn2023quantum}
G.~J. Milburn.
\newblock ``Quantum learning machines''~(2023).
\newblock  \href{http://arxiv.org/abs/2305.07801}{arXiv:2305.07801}.

\bibitem{Stil07b}
S.~Still, J.~P. Crutchfield, and C.~J. Ellison.
\newblock ``Optimal causal inference: {Estimating} stored information and
  approximating causal architecture''.
\newblock \href{https://dx.doi.org/https://doi.org/10.1063/1.3489885}{CHAOS
  {\bf 20}, 037111}~(2010).

\bibitem{creutzig2009past}
Felix Creutzig, Amir Globerson, and Naftali Tishby.
\newblock ``Past-future information bottleneck in dynamical systems''.
\newblock
  \href{https://dx.doi.org/https://doi.org/10.1103/PhysRevE.79.041925}{Physical
  Review E {\bf 79}, 041925}~(2009).

\bibitem{marzen2016predictive}
Sarah~E Marzen and James~P Crutchfield.
\newblock ``Predictive rate-distortion for infinite-order markov processes''.
\newblock
  \href{https://dx.doi.org/https://doi.org/10.1007/s10955-021-02698-1}{Journal
  of Statistical Physics {\bf 163}, 1312--1338}~(2016).

\bibitem{hastie2009elements}
Trevor Hastie, Robert Tibshirani, Jerome~H Friedman, and Jerome~H Friedman.
\newblock ``The elements of statistical learning: data mining, inference, and
  prediction''.
\newblock
  \href{https://dx.doi.org/https://doi.org/10.1007/978-0-387-21606-5}{Volume~2}.
\newblock Springer. ~(2009).

\bibitem{srivastava2014dropout}
Nitish Srivastava, Geoffrey Hinton, Alex Krizhevsky, Ilya Sutskever, and Ruslan
  Salakhutdinov.
\newblock ``Dropout: a simple way to prevent neural networks from
  overfitting''.
\newblock The journal of machine learning research {\bf 15}, 1929--1958~(2014).
\newblock  url:~\url{http://jmlr.org/papers/v15/srivastava14a.html}.

\bibitem{Boyd16c}
A.~B. Boyd, D.~Mandal, and J.~P. Crutchfield.
\newblock ``Correlation-powered information engines and the thermodynamics of
  self-correction''.
\newblock \href{https://dx.doi.org/10.1103/PhysRevE.95.012152}{Physical Review
  E {\bf 95}, 012152}~(2017).

\bibitem{jaynes2003probability}
Edwin~T Jaynes.
\newblock ``Probability theory: The logic of science''.
\newblock \href{https://dx.doi.org/https://doi.org/10.1063/1.1825273}{Cambridge
  university press}. ~(2003).

\bibitem{grunwald2007minimum}
Peter~D Gr{\"u}nwald.
\newblock ``The minimum description length principle''.
\newblock
  \href{https://dx.doi.org/https://doi.org/10.7551/mitpress/4643.001.0001}{MIT
  press}. ~(2007).

\bibitem{Jun14a}
Y.~Jun, M.~Gavrilov, and J.~Bechhoefer.
\newblock ``High-precision test of {Landauer's} principle in a feedback trap''.
\newblock
  \href{https://dx.doi.org/https://doi.org/10.1103/PhysRevLett.113.190601}{Physical
  Review Letters {\bf 113}, 190601}~(2014).

\bibitem{Boyd16e}
A.~B. Boyd, D.~Mandal, P.~M. Riechers, and J.~P. Crutchfield.
\newblock ``Transient dissipation and structural costs of physical information
  transduction''.
\newblock
  \href{https://dx.doi.org/https://doi.org/10.1103/PhysRevLett.118.220602}{Physical
  Review Letters {\bf 118}, 220602}~(2017).

\bibitem{Boyd16d}
A.~B. Boyd, D.~Mandal, and J.~P. Crutchfield.
\newblock ``Leveraging environmental correlations: The thermodynamics of
  requisite variety''.
\newblock \href{https://dx.doi.org/10.1007/s10955-017-1776-0}{Journal of
  Statistical Physics {\bf 167}, 1555--1585}~(2016).

\bibitem{radaelli2023fisher}
Marco Radaelli, Gabriel~T Landi, Kavan Modi, and Felix~C Binder.
\newblock ``Fisher information of correlated stochastic processes''.
\newblock \href{https://dx.doi.org/10.1088/1367-2630/acd321}{New Journal of
  Physics {\bf 25}, 053037}~(2023).

\bibitem{Riech23a}
Paul~M Riechers.
\newblock ``Ultimate limit on learning non-markovian behavior: Fisher
  information rate and excess information''~(2023).
\newblock  \href{http://arxiv.org/abs/2310.03968}{arXiv:2310.03968}.

\bibitem{lehmann2006theory}
Erich~L Lehmann and George Casella.
\newblock ``Theory of point estimation''.
\newblock \href{https://dx.doi.org/https://doi.org/10.1007/b98854}{Springer
  Science \& Business Media}. ~(2006).

\bibitem{hsu2023strange}
Alexander Hsu and Sarah~E Marzen.
\newblock ``Strange properties of linear reservoirs in the infinitely large
  limit for prediction of continuous-time signals''.
\newblock
  \href{https://dx.doi.org/https://doi.org/10.1007/s10955-022-03040-z}{Journal
  of Statistical Physics {\bf 190}, 32}~(2023).

\bibitem{carroll2020reservoir}
Thomas~L Carroll.
\newblock ``Do reservoir computers work best at the edge of chaos?''.
\newblock \href{https://dx.doi.org/https://doi.org/10.1063/5.0038163}{Chaos: An
  Interdisciplinary Journal of Nonlinear Science {\bf 30}, 121109}~(2020).

\bibitem{schurmann2004edge}
Felix Sch{\"u}rmann, Karlheinz Meier, and Johannes Schemmel.
\newblock ``Edge of chaos computation in mixed-mode vlsi-a hard liquid''.
\newblock Advances in neural information processing systems{\bf 17}~(2004).
\newblock
  url:~\url{proceedings.neurips.cc/paper_files/paper/2004/file/dbab2adc8f9d078009ee3fa810bea142-Paper.pdf}.

\bibitem{marzen2018infinitely}
Sarah Marzen.
\newblock ``Infinitely large, randomly wired sensors cannot predict their input
  unless they are close to deterministic''.
\newblock \href{https://dx.doi.org/10.1371/journal.pone.0202333}{Plos one {\bf
  13}, e0202333}~(2018).

\bibitem{pascanu2013difficulty}
Razvan Pascanu, Tomas Mikolov, and Yoshua Bengio.
\newblock ``On the difficulty of training recurrent neural networks''.
\newblock In Sanjoy Dasgupta and David McAllester, editors, Proceedings of the
  30th International Conference on Machine Learning.
\newblock Volume~28 of Proceedings of Machine Learning Research, pages
  1310--1318.
\newblock PMLR~(2013).

\bibitem{zhang2019learning}
Amy Zhang, Zachary~C Lipton, Luis Pineda, Kamyar Azizzadenesheli, Anima
  Anandkumar, Laurent Itti, Joelle Pineau, and Tommaso Furlanello.
\newblock ``Learning causal state representations of partially observable
  environments''~(2019).
\newblock  \href{http://arxiv.org/abs/1906.10437}{arXiv:1906.10437}.

\bibitem{kim2021code}
Seohyun Kim, Jinman Zhao, Yuchi Tian, and Satish Chandra.
\newblock ``Code prediction by feeding trees to transformers''.
\newblock In 2021 IEEE/ACM 43rd International Conference on Software
  Engineering (ICSE).
\newblock \href{https://dx.doi.org/10.1109/ICSE43902.2021.00026}{Pages
  150--162}.
\newblock IEEE~(2021).

\bibitem{o2021context}
Joe O'Connor and Jacob Andreas.
\newblock ``What context features can transformer language models
  use?''~(2021).
\newblock  \href{http://arxiv.org/abs/2106.08367}{arXiv:2106.08367}.

\bibitem{gu2012quantum}
Mile Gu, Karoline Wiesner, Elisabeth Rieper, and Vlatko Vedral.
\newblock ``Quantum mechanics can reduce the complexity of classical models''.
\newblock \href{https://dx.doi.org/https://doi.org/10.1038/ncomms1761}{Nature
  communications {\bf 3}, 762}~(2012).

\bibitem{takaki2021learning}
Yuto Takaki, Kosuke Mitarai, Makoto Negoro, Keisuke Fujii, and Masahiro
  Kitagawa.
\newblock ``Learning temporal data with a variational quantum recurrent neural
  network''.
\newblock
  \href{https://dx.doi.org/https://doi.org/10.1103/PhysRevA.103.052414}{Physical
  Review A {\bf 103}, 052414}~(2021).

\bibitem{riechers2021}
Paul~M. Riechers and Mile Gu.
\newblock ``Initial-state dependence of thermodynamic dissipation for any
  quantum process''.
\newblock \href{https://dx.doi.org/10.1103/PhysRevE.103.042145}{Phys. Rev. E
  {\bf 103}, 042145}~(2021).

\bibitem{binder2018practical}
Felix~C Binder, Jayne Thompson, and Mile Gu.
\newblock ``Practical unitary simulator for non-markovian complex processes''.
\newblock \href{https://dx.doi.org/10.1103/PhysRevLett.120.240502}{Physical
  review letters {\bf 120}, 240502}~(2018).

\bibitem{elliott2022quantum}
Thomas~J Elliott, Mile Gu, Andrew~JP Garner, and Jayne Thompson.
\newblock ``Quantum adaptive agents with efficient long-term memories''.
\newblock
  \href{https://dx.doi.org/https://doi.org/10.1103/PhysRevX.12.011007}{Physical
  Review X {\bf 12}, 011007}~(2022).

\bibitem{yano2021efficient}
Hiroshi Yano, Yudai Suzuki, Kohei~M Itoh, Rudy Raymond, and Naoki Yamamoto.
\newblock ``Efficient discrete feature encoding for variational quantum
  classifier''.
\newblock \href{https://dx.doi.org/10.1109/TQE.2021.3103050}{IEEE Transactions
  on Quantum Engineering {\bf 2}, 1--14}~(2021).

\bibitem{loomis2019strong}
Samuel~P Loomis and James~P Crutchfield.
\newblock ``Strong and weak optimizations in classical and quantum models of
  stochastic processes''.
\newblock
  \href{https://dx.doi.org/https://doi.org/10.1007/s10955-019-02344-x}{Journal
  of Statistical Physics {\bf 176}, 1317--1342}~(2019).

\bibitem{loomis2020thermal}
Samuel~P Loomis and James~P Crutchfield.
\newblock ``Thermal efficiency of quantum memory compression''.
\newblock
  \href{https://dx.doi.org/https://doi.org/10.1103/PhysRevLett.125.020601}{Physical
  review letters {\bf 125}, 020601}~(2020).

\bibitem{banchi2021generalization}
Leonardo Banchi, Jason Pereira, and Stefano Pirandola.
\newblock ``Generalization in quantum machine learning: A quantum information
  standpoint''.
\newblock
  \href{https://dx.doi.org/https://doi.org/10.1103/PRXQuantum.2.040321}{PRX
  Quantum {\bf 2}, 040321}~(2021).

\bibitem{Crut08b}
C.~J. Ellison, J.~R. Mahoney, and J.~P. Crutchfield.
\newblock ``Prediction, retrodiction, and the amount of information stored in
  the present''.
\newblock
  \href{https://dx.doi.org/https://doi.org/10.1007/s10955-009-9808-z}{Journal
  of Statistical Physics {\bf 136}, 1005--1034}~(2009).

\bibitem{riechers2020balancing}
Paul~M Riechers, Alexander~B Boyd, Gregory~W Wimsatt, and James~P Crutchfield.
\newblock ``Balancing error and dissipation in computing''.
\newblock
  \href{https://dx.doi.org/https://doi.org/10.1103/PhysRevResearch.2.033524}{Physical
  Review Research {\bf 2}, 033524}~(2020).

\bibitem{wimsatt2022trajectory}
Gregory Wimsatt, Alexander~B Boyd, and James~P Crutchfield.
\newblock ``Trajectory class fluctuation theorem''~(2022).
\newblock  \href{http://arxiv.org/abs/2207.03612}{arXiv:2207.03612}.

\bibitem{Cove06a}
T.~M. Cover and J.~A. Thomas.
\newblock ``Elements of information theory''.
\newblock \href{https://dx.doi.org/10.1002/047174882X}{Wiley-Interscience}. New
  York~(2006).
\newblock Second edition.

\bibitem{crutchfield2016not}
James~P Crutchfield and Cina Aghamohammadi.
\newblock ``Not all fluctuations are created equal: Spontaneous variations in
  thermodynamic function''~(2016).
\newblock  \href{http://arxiv.org/abs/1609.02519}{arXiv:1609.02519}.

\end{thebibliography}

\end{document}